\documentclass[journal]{IEEEtran}
\usepackage{subcaption}
\usepackage{tabularx}
\usepackage{makecell}
\usepackage{multicol,multirow}
\usepackage{array} 
\usepackage{cite}
\usepackage{tikz}
\usepackage{pgfplots}
\usepackage{standalone}
\usepackage{textcomp}
\usepackage{amssymb,amsfonts}
\usepackage[cmex10]{amsmath}
\usepackage{amsthm}
\usepackage{algorithmic,tcolorbox}
\usepackage{algorithm}
\usepackage{graphicx}
\usepackage{bbding}
\usepackage{textcomp}
\usepackage{threeparttable}
\usepackage{framed}
\usepackage{verbatim}
\usepackage{bbm}
\usepackage{setspace}
\usepackage{mathrsfs}
\usepackage{bm}
\usepackage{url}
\usepackage{multirow}
\usepackage{verbatim}
\usepackage{float}
\usepackage{hyperref}
\usepackage{nomencl}
\allowdisplaybreaks
\makenomenclature

\renewcommand{\arraystretch}{1}

\def\blue{\color{black}}

\newtheorem{theorem}{Theorem}[section]
\newtheorem{lemma}[theorem]{Lemma}
\newtheorem{proposition}[theorem]{Proposition}
\newtheorem{corollary}[theorem]{Corollary}
\newtheorem{definition}[theorem]{Definition}
\newtheorem{remark}[theorem]{Remark}

\newtheorem{assumption}[theorem]{Assumption}
\usepackage{pgfplotstable}
\pgfplotsset{
    legend image code/.code={
        \draw [#1] (0cm,-0.1cm) rectangle (0.6cm,0.1cm);
    },
}
\def\BibTeX{{\rm B\kern-.05em{\sc i\kern-.025em b}\kern-.08em
    T\kern-.1667em\lower.7ex\hbox{E}\kern-.125emX}}
\pgfplotsset{compat=1.15}

\renewcommand{\baselinestretch}{0.98}

\title{\Huge Consumer-based Carbon Costs: \\ Integrating Consumer Carbon Preferences in Electricity Markets}

\author{Wenqian Jiang, Aditya Rangarajan, \emph{and} Line Roald
}

\begin{document}

\maketitle

\begin{abstract}
An increasing share of consumers care about the carbon footprint of their electricity. This paper analyzes a method to integrate consumer carbon preferences in the electricity market-clearing by introducing consumer-based carbon costs and a carbon allocation mechanism. Specifically, consumers submit not only bids for power but also assign a cost to the carbon emissions incurred by their electricity use. The carbon allocation mechanism then assigns emissions from generation to consumers to minimize overall carbon costs. 
Our analysis starts from a previously proposed centralized market clearing formulation that maximizes social welfare under consideration of generation costs, consumer utility, and consumer carbon costs. 
We then derive an equivalent equilibrium formulation that incorporates a carbon allocation problem and gives rise to a set of carbon-adjusted electricity prices for both consumers and generators. We prove that the carbon-adjusted prices are higher for low-emitting generators and consumers with high carbon costs. Further, we prove that this new paradigm satisfies the same desirable market properties as standard electricity markets based on locational marginal prices, namely revenue adequacy and individual rationality, and demonstrate that a carbon tax on generators is equivalent to imposing a uniform carbon cost on consumers. Using a simplified three-bus system and the RTS-GMLC system, we illustrate that consumer-based carbon costs contribute to greener electricity market clearing both through generation redispatch and demand reductions.

\end{abstract}


\section{Introduction}
The electricity sector contributes approximately 30\% of total energy-related emissions \cite{renewables2022analysis}, making it a key target for decarbonization efforts 
through the adoption of low-carbon generation technologies. 
The availability of low-carbon electricity varies across time and space based on the regional generation mix and weather conditions, leading to spatio-temporal variations in the carbon footprint of electricity generation. These variations have inspired a growing group of \emph{carbon-sensitive} electricity consumers to adapt when (and in some cases where) they consume electricity to reduce carbon emissions from their electricity usage. 
Examples of such carbon-sensitive consumers range from residential customers to large corporations, such as hyperscale computing companies or producers of clean hydrogen.

Carbon-sensitive consumers may be willing to pay a premium for low-carbon electricity or, conversely, be less willing to pay for electricity from polluting generators.
However, current electricity markets minimize cost or maximize social welfare without explicitly accounting for emissions, potentially leading to 
economically efficient but environmentally suboptimal outcomes. 
Although certain carbon pricing mechanisms, such as the European Union's Emissions Trading System \cite{EUETs} and California's Cap-and-Trade Program \cite{CCTp}, have been implemented, they offer limited avenues for consumers to actively express their individual preferences for cleaner electricity. Instead, carbon-sensitive consumers must rely on indirect methods, such as those outlined in the Greenhouse Gas (GHG) Protocol \cite{GHGp2},
to estimate and mitigate their electricity-related emissions. Crucially, these approaches remain decoupled from real-time electricity markets, which limits their influence on real-time generation dispatch decisions and the resulting carbon emissions. 

To address the current disconnect between carbon accounting and electricity market clearing, we propose an electricity market clearing paradigm that inherently accounts for consumer carbon preferences and incorporates a carbon allocation mechanism to assign emissions from generators to consumers.

\subsection{Related Works}
To mitigate carbon emissions in power systems, researchers and policy-
makers have considered emissions trading schemes \cite{EUETs, CCTp} and carbon taxes on generators \cite{baranzini2000future,fischer2008environmental}.  
These mechanisms increase the cost of carbon-intensive generation, thus dispatching more low-carbon generators and
reducing system-wide emissions \cite{shukla2022climate,change2022mitigating}.
However, the increased generation costs are ultimately paid by consumers via elevated electricity prices \cite{nazifi2021carbon}. 
An analysis of the European emissions trading system showed that emission reductions primarily arise from reductions in demand due to higher electricity prices, rather than increased dispatch of clean generation \cite{chen2008implications, sijm2005co2}. 
Further, a uniform carbon tax or an emission trading scheme cannot account for differences in individual consumers' preference for low-carbon power.

{\blue 
As outlined above, the rise of carbon-sensitive consumers motivates a discussion on how to most effectively empower consumers to express and act on such preferences. 
However, while counting the carbon emissions of generators is a fairly straightforward task, assessing (and potentially penalizing) the carbon emissions of individual consumers requires allocating emissions from individual generators to individual consumers. How to do this in a fair and transparent manner is less obvious
and subject to ongoing debate. To outline these challenges, we review three parts of the literature related to carbon accounting and carbon emission reductions of electricity consumers, namely (1) carbon accounting under the GHG Protocol, (2) carbon-aware load shifting, and (3) proposals to integrate carbon allocation and carbon constraints on consumers within the electricity market clearing. We also note that the paper and our review are focused on the emissions associated with real-time electricity consumption, rather than the long-term evolution of the system. 
} 

\subsubsection{Carbon accounting under the GHG Protocol}
The GHG Protocol represents the industry standard for carbon accounting of electricity usage.
The GHG Protocol scope 2 guidance \cite{GHGp2} provides two mechanisms for carbon accounting of electricity use: market-based accounting and location-based accounting. With \emph{market-based accounting}, consumers can purchase renewable energy certificates (RECs) \cite{wiser1998green,jensen2002interactions,USEP}, often coupled with power purchase agreements (PPAs) \cite{facebook, Googleppa, Microsoftppa}, to claim that they are carbon-free. 
{\blue This method decouples physical electricity consumption from carbon emissions accounting, as consumers may utilize carbon-intensive electricity without punishment as long as they procure a sufficient amount of RECs. } 
With \emph{location-based accounting}, consumers calculate their CO$_2$ emissions after-the-fact by multiplying their electricity consumption with the average CO$_2$ emissions in the grid they are connected to. {\blue The average carbon emissions are typically calculated for the entire region and averaged in time (e.g. average yearly emissions), 
thus failing to account for the impact of consumers' specific locations and time of electricity consumption.
Ongoing updates to the GHG protocol seek to establish a closer relationship between physical consumption and carbon accounting, e.g. by implementing stricter requirements on temporal and geographical matching for RECs. However, carbon accounting under the GHG protocol would remain disconnected from the real-time electricity market clearing, giving rise to inefficiencies in adjusting consumers' load patterns.}

\subsubsection{Carbon-Aware Load Shifting}
Beyond emissions accounting, there has been significant interest in carbon-aware load shifting, where consumers use real-time carbon emission metrics to shift consumption to times or locations with a lower carbon footprint.
Carbon-aware load shifting has been studied in the context of data centers\cite{lindberg2020environmental, zheng2020mitigating, wiesner2021let, radovanovic2022carbon, acun2023carbon, cao2022toward}, hydrogen production \cite{ricks2023minimizing, zhang2020flexible}, and residential and commercial electricity usage \cite{present2024choosing, vader2022operational,goldstein2020carbon, tabors2021methodology}. 
Recently, the increasing availability of real-world carbon emissions data, provided directly by both grid
operators \cite{gridoperater1, gridoperater2} and third-party organizations \cite{carboncom2,electricitymap}, has enabled real-world implementations of carbon-aware load shifting. 
However, the choice of carbon intensity metric remains a key challenge to carbon-aware load shifting. Existing literature has proposed \cite{lindberg2020environmental, chen2024contributions, chen2024enhance} and 
compared \cite{lindberg2021guide, lindberg2022using, sukprasert2024implications, gorka2025electricityemissions} the impacts of choosing different carbon emissions metrics to guide the load shifting and also studied the effect of transmission constraints \cite{sofia2024carbon}. 
These studies have found that average carbon emissions, the most commonly used metric, may have a limited or negative impact on emissions \cite{gorka2025electricityemissions, jiang2025can}. Locational marginal emission, an alternative metric, is more effective for (small) load shifts, but it is highly volatile, and thus, not necessarily suitable for carbon accounting \cite{gorka2025electricityemissions}. 
However, carbon-aware load shifting has some fundamental limitations. 
In practice, load shifting would usually happen after the market has cleared, see, e.g., \cite{lindberg2022using, lindberg2021guide,gorka2025electricityemissions}, which may lead to suboptimal outcomes because consumers respond to outdated carbon signals from the previous market clearing. 
Further, load shifting can only impact the marginal generator and cannot cause a redispatch due to changes in the merit order the same way that, e.g. a carbon tax does.

\subsubsection{Carbon-Aware Electricity Market Clearing}
A more comprehensive approach to account for consumers' carbon preferences is to explicitly incorporate them into the electricity market. 
Recently, there have been a few different proposals modeling this. 
One strain of work models market clearing with carbon-aware consumers as a bilevel optimization problem, 
where either carbon-sensitive consumers adapt their consumption to minimize emissions, assuming full knowledge of the electricity market clearing problem of the ISO \cite{lindberg2022using}, or the ISO minimizes cost, assuming full information about the operating strategies, behaviors, and carbon preferences of carbon-sensitive consumers \cite{chen2024enhance}. 
However, neither of those two setups is practical due to the need for extensive information sharing.

Another strain of work has proposed direct integration of a carbon allocation mechanism in the market clearing itself. In \cite{chen2024carbon}, a carbon allocation mechanism based on the  carbon flow method \cite{kang2015carbon}, \textcolor{black}{is integrated into the market clearing, and consumers express their carbon preferences by defining explicit limits on their allowed carbon emissions.} 
One challenge of this method is that it is not clear if the proportional sharing principle \cite{kirschen1997contributions, bialek1996tracing}, which is essential to the carbon flow method, is the “right" definition of power and carbon flow tracing. Furthermore, it can be very challenging for consumers or the independent system operator (ISO) to define what the values of the carbon emission cap should be.

These drawbacks were addressed by our recent work in \cite{jiang2025greening}, which allows consumers to state their carbon preferences in the form of a carbon cost instead of an explicit cap \textcolor{black}{and proposes a less restrictive carbon allocation mechanism that minimizes the overall carbon cost, rather than adopting a (somewhat arbitrary) physical tracing approach}.
However, 
any model that deviates from the pure cost minimization of the standard market clearing will increase the cost of generation and system operation. Ref. \cite{jiang2025greening} does not explore how consideration of carbon costs would impact electricity price formation, who would pay for possible increases in generation cost, and whether the proposed market clearing still would satisfy desirable properties of existing electricity market clearing, such as revenue adequacy and individual rationality.

\subsection{Our Contributions}
This paper seeks to address these gaps by analyzing the proposed model in \cite{jiang2025greening} 
using
equilibrium modeling, which
has been widely applied in the analysis of existing electricity markets \cite{ventosa2005electricity, hobbs2007nash,gabriel2012complementarity}.
Our main contributions are fourfold: 
\\\emph{1)}
By comparing optimality conditions, we show that the proposed centralized model in \cite{jiang2025greening} gives rise to an \emph{equivalent equilibrium model} representing profit maximization problems of generators, consumers, and transmission owners; the price setting problem of the system operator; and an optimization problem solved by a carbon manager, who allocates power (and associated emissions) from generators to loads to minimize total carbon costs. 
\\\emph{2)} Based on this equilibrium model, we define a set of \emph{carbon-adjusted prices}, which serve as coordinating variables across the profit maximization problems of generators and consumers, leading to outcomes that are consistent with the centralized model. 
We prove that the carbon adjustments to the prices are such that consumers who submit higher carbon costs pay comparatively more for their electricity, while generators with lower emissions are paid comparatively more.
\\\emph{3)} We demonstrate that the proposed model satisfies similar \emph{desirable market properties} as current standard (carbon-agnostic) markets, namely revenue adequacy and individual rationality. 
\\\emph{4)} We show that current electricity market clearing with or without uniform carbon taxes on generators can be interpreted as special cases of our model.

These results are illustrated in case studies on a small three-bus system and the RTS-GMLC system \cite{barrows2019ieee}.


The remainder of the paper is organized as follows.  Section \ref{sec2} reviews our proposed market clearing model with consumer-based carbon costs from \cite{jiang2025greening}. Section \ref{sec3} derives its equivalent equilibrium formulation, while Section \ref{sec4} introduces and analyzes the carbon-adjusted electricity prices. Section \ref{secmp} investigates market properties of our model, and Section \ref{secmg} discusses its relationship to the standard (carbon-agnostic) model and a model with uniform carbon tax. Section \ref{secns} provides a numerical case study to illustrate theoretical results. Finally, Section \ref{sec5} concludes the paper. 

\section{Electricity Market Clearing with {\blue Consumer-Based} Carbon Cost}
\label{sec2}
Our recent work \cite{jiang2025greening} proposed a new model for electricity market clearing, which incorporates \emph{{\blue consumer-based} carbon costs} and a \emph{carbon-allocation mechanism}. The carbon costs allow consumers to define the cost they associate with carbon emissions from electricity usage.
The carbon allocation mechanism allocates power from generators to consumers, prioritizing the allocation of low-carbon power to consumers with high carbon costs. 
This model provides a new opportunity for consumers to express their carbon preferences in the electricity market clearing. 
We outline our model here and refer the reader to \cite{jiang2025greening} for more details.



We consider an electric power network with the set of buses, consumers, transmission lines, and generators denoted by $\mathcal{N}$, $\mathcal{D}$, $\mathcal{L}$ and $\mathcal{G}$, respectively. Let $\mathcal{G}_i\subset \mathcal{G}$ and $\mathcal{D}_i\subset \mathcal{D}$ be the subset of generators and loads connected to bus $i$, and $(i,j)\in\mathcal{L}$ denote the transmission line from bus $i$ to bus $j$. For notational clarity, we use subscripts $G$ and $D$ to differentiate variables or parameters for generators and consumers. 
Prior to the market clearing, each generator $g\in \mathcal{G}$ submits their generation costs $c_{G,g}$, maximum and minimum generation capacities $P_{G,g}^{\max}$ and $P_{G,g}^{\min}$ and emission factors $e_{G,g}$. Each consumer $d\in \mathcal{D}$ submits their maximum and minimum demand $P_{D,d}^{\max}$, and $P_{D,d}^{\min}$ and bids for electricity consumption $u_{D,d}$, reflecting the utility (or revenue) they derive from consuming electricity. Consumers also submit their carbon cost $c_{D,d}$, given in units of [\$/t$\rm{CO}_2$]. 
This cost may be directly tied to concrete costs such as \textcolor{black}{the purchase of renewable energy certificates or} carbon emission penalties from carbon taxes or cap-and-trade schemes\textcolor{black}{.
Alternatively, they may reflect} an internally defined “carbon cost”, quantifying how much revenue the consumer is willing to forgo to avoid carbon emissions{\blue 
\footnote{For example, if a consumer is willing to pay \$100/MWh for zero-carbon electricity but only \$50/MWh for electricity emitting 1 tCO$_2$/MWh, the implied carbon cost is \$50/tCO$_2$.}. 
While we acknowledge that there may be practical challenges associated with eliciting and verifying these carbon costs, both for the individual user and from a market power mitigation perspective, a thorough analysis of this is reserved for future work. }



Once bids for generation, consumption, and carbon are known, the ISO solves the problem (\ref{eq1coo}). This problem is a modified version of the DC optimal power flow (DC OPF) problem, where we expand the objective function to consider consumers' carbon costs and model the carbon allocation mechanism through a set of additional constraints that assign generated power (and associated emissions) from generators to loads. The problem is as follows:
\begin{subequations}
    \label{eq1coo}
    \begin{align}
    ~~&\max_{P_G, P_D, \theta, \pi, E_D}u_D^\intercal P_D-c_D^\intercal &&\!\!\!\!\!\!\!\!\!E_D -c_G^\intercal P_G\label{eq1coobj}\\
    \text{s.t.~~}
    & \sum_{d\in \mathcal{D}_i}P_{D,d}+\sum_{j:(i,j)\in \mathcal{L}}\beta_{ij}(\theta_i&&\!\!\!\!\!\!\!-\theta_j)=\sum_{g\in \mathcal{G}_i}P_{G,g},\notag\\
    &\qquad\qquad\qquad\qquad\qquad &&\!\!\!\!\!\!\!\!\!\!\!\forall i \in \mathcal{N}, \label{eq1coa}\quad\quad~:~\lambda_{P,i}\\
    &\beta_{ij}(\theta_i-\theta_j)\leq F_{ij}^{\rm{lim}},\quad&&\!\!\!\!\!\!\!\!\!\!\!\forall (i,j)\in \mathcal{L}, \quad:~\overline{\eta}_{L,ij}\label{eq1cob}\\ 
    &\beta_{ij}(\theta_i-\theta_j)\ge -F_{ij}^{\rm{lim}},\quad&&\!\!\!\!\!\!\!\!\!\!\!\forall (i,j)\in \mathcal{L},\quad:~ \underline{\eta}_{L,ij}\label{eq1coblim}\\
    &P_{G,g}^{\min}\leq P_{G,g}\leq P_{G,g}^{\max}, \quad&&\!\!\!\!\!\!\!\!\!\!\!\forall g \in \mathcal{G},\label{eq1oc}\quad :~\overline{\eta}_{G,g}, \underline{\eta}_{G,g}\\
    &P_{D,d}^{\min}\leq P_{D,d}\leq P_{D,d}^{\max},\quad&&\!\!\!\!\!\!\!\!\!\!\!\forall d \in \mathcal{D}, \label{eq1od}\quad:~\overline{\eta}_{D,d}, \underline{\eta}_{D,d}\\
    &\theta_{\rm{ref}} = 0, &&\!\!\!\!\!\!\!\!\!\!\!\label{eq1coe}\\
    & \sum_{d\in \mathcal{D}}\pi_{g,d} = P_{G,g},\ &&\!\!\!\!\!\!\!\!\!\!\!\forall g \in \mathcal{G},\quad:~\lambda_{G,g}  \label{eq1cof}\\
    & \sum_{g\in\mathcal{G}}\pi_{g,d} = P_{D,d},\ &&\!\!\!\!\!\!\!\!\!\!\!\forall d \in \mathcal{D},\quad :~\lambda_{D,d}\label{eq1cog}\\
    & \sum_{g\in\mathcal{G}}e_{G,g}\pi_{g,d} = E_{D,d}, \ 
    &&\!\!\!\!\!\!\!\!\!\!\!\forall d \in \mathcal{D}, \quad :~\lambda_{E,d}\label{eq1cho}\\
    & \pi_{g,d}\geq 0,\ &&\!\!\!\!\!\!\!\!\!\!\!\forall g \in \mathcal{G}, \ \forall d \in \mathcal{D} \label{eq1cooi}.
\end{align}
\end{subequations}
Here, the optimization variables are the generation dispatch $P_G = \{P_{G,g}: g \in \mathcal{G}\}$, the voltage angle $\theta = \{\theta_i: i\in \mathcal{N}\}$, the flexible load $P_D = \{P_{D,d}:d \in \mathcal{D}\}$, the generation-load allocation matrix $\pi=\{\pi_{g,d}:g \in \mathcal{G}, d \in \mathcal{D}\}$ reflecting the amount of power assigned from each generator to each load, and the total carbon emission for each consumer $E_D = \{E_{D,d}:d\in \mathcal{D}\}$. We describe each part below.

\subsubsection*{Carbon-aware objective function} The objective function (\ref{eq1coobj}) maximizes carbon-aware social welfare which includes the utility term, carbon cost term, and generation cost term. 


\subsubsection*{DC OPF constraints} The constraints (\ref{eq1coa})-(\ref{eq1coe}) are similar to those of the standard DC OPF.
Constraint (\ref{eq1coa}) ensures that nodal power balance constraints are met, with $\beta_{ij}\in \mathbb{R}$ denoting the susceptance value of the transmission line $(i,j)$ from bus $i$ to bus $j$. Constraints (\ref{eq1cob}) and (\ref{eq1coblim}) are the transmission line limits, where $F_{ij}^{\rm{lim}}$ represents the transmission capacity, which we assume is the same in both directions. Constraints (\ref{eq1oc}) and (\ref{eq1od}) enforce the limits on generation capacity and demand flexibility, while constraint (\ref{eq1coe}) sets the voltage angle at the reference bus to zero. The variables $\lambda_{P,i}, \overline{\eta}_{L,ij}, \underline{\eta}_{L,ij}, \overline{\eta}_{G,g}, \underline{\eta}_{G,g},\overline{\eta}_{D,d}, \underline{\eta}_{D,d}, \lambda_{G,g}, \lambda_{D,d}, \lambda_{E,d}$ after the colon at each constraint represent dual variables (or Lagrange multipliers) for corresponding constraints. 

\subsubsection*{Carbon allocation mechanism}
The remaining constraints (\ref{eq1cof})-(\ref{eq1cooi}) 
assign a portion of the power, and associated emissions, from each generator $g$ to each load $d$, represented by the power allocation $\pi_{g,d}$. 
Constraint (\ref{eq1cof}) ensures that the total amount of power allocated from the generator $g$ to all the loads $d\in \mathcal{D}$ equals the actual power dispatched from this generator, while (\ref{eq1cog}) enforces that the sum of power allocated to a given load is equal to its total power consumption. 
Constraint (\ref{eq1cho}) then computes the total emission for each consumer $E_{D,d}$ based on emission factors and the amount of power obtained from different generators.
Constraint (\ref{eq1cooi}) ensures that all allocations are non-negative, which guarantees that all loads will have non-negative emissions assuming non-negative generator emission factors.

We note that this carbon allocation mechanism itself assumes that any generated power can be assigned to any load and does not explicitly consider physical constraints such as grid topology, power flow patterns or congestion in the system. These physical characteristics of grid operations are still accounted for by the DC power flow constraints \eqref{eq1coa}-\eqref{eq1coblim}. 

\section{Equivalent Equilibrium Formulation}
\label{sec3}


Problem (\ref{eq1coo}) is a centralized, system-level model, where a system operator gathers information about costs, carbon emissions, and availability from generators and consumers and clears the market with the goal of maximizing carbon-aware social welfare. However, it is not clear whether the proposed model satisfies desired market properties, e.g., produces prices that incentivize individual actors to comply with the market outcome. To enable analysis of these aspects, we next show that there exists an equivalent equilibrium problem that represents the optimization problems solved by individual actors in the market, connected through a set of coordinating variables. 

Our analysis is motivated by the equilibrium modeling of current (carbon-agnostic) electricity markets based on locational marginal pricing \cite{ventosa2005electricity,hobbs2007nash}. Before diving into more details about the derivation of equivalent equilibrium model of our carbon-aware market clearing model, we briefly summarize some key points established in the analysis of existing markets. 
Equilibrium models of existing markets typically consider 
generators, consumers, transmission owners, and the ISO as participants. 
The 
generators, consumers, and transmission owners aim to maximize their profits given electricity prices as input parameters, while the ISO determines the market clearing price. 
It is commonly assumed that the participants are, between them, playing a noncooperative game, and thus the optimal solution to this game is defined as a Nash equilibrium, corresponding to a situation where no participant can improve their outcomes by unilaterally changing their decisions. Under a price-taking assumption (no strategic bidding or market power), the equilibrium problem is equivalent to a standard market-clearing problem, i.e., the optimality conditions of the two problems are the same. Therefore, solutions from the equilibrium problem are the same as those of the single central problem and thus also maximize social welfare. The equilibrium problem can be rewritten as a mixed complementarity problem; 
more details, examples, and formulations can be found in \cite{ferris1997engineering,gabriel2012complementarity, ferris2025optimizing}.

Inspired by prior results for standard electricity markets, we seek to establish an equivalent equilibrium formulation for the centralized problem (\ref{eq1coo}). To achieve this, we first describe the optimality conditions of Problem (\ref{eq1coo}). Using these optimality conditions, we then define an equilibrium problem that includes profit maximization for generators, consumers, and transmission owners; the price-setting problem of the ISO; and an optimization problem solved by a carbon
manager, who allocates carbon emissions from generation to loads. 
\subsection{Optimality Conditions}
\label{optcon}
\subsubsection{Dual problem} To obtain the optimality conditions, we consider the dual problem of (\ref{eq1coo}), which, based on convex optimization theory \cite{boyd2004convex}, is given as follows:
\begin{subequations}
\label{dualp}
    \begin{align}
        \min_{\lambda,\eta}~~ &\sum_{g\in \mathcal{G}}\overline{\eta}_{G,g}P_{G,g}^{\max}-\sum_{g \in \mathcal{G}}\underline{\eta}_{G,g}P_{G,g}^{\min}+\sum_{d\in \mathcal{D}}\overline{\eta}_{D,d}P_{D,d}^{\max},\notag\\
        &-\sum_{d\in \mathcal{D}}\underline{\eta}_{D,d}P_{D,d}^{\min}+\sum_{(i,j)\in \mathcal{L}}(\overline{\eta}_{L,ij}+\underline{\eta}_{L,ij})F_{i,j}^{\lim}\\
        \text{s.t.~~} &
    -c_{G,g}+\lambda_{P,i:g\in\mathcal{G}_i} -\overline{\eta}_{G,g}+\underline{\eta}_{G,g}+\lambda_{G,g}=0,\notag\\
    &\qquad\qquad\qquad\qquad\qquad \forall g\in \mathcal{G},~~:~P_{G,g} \label{s1}\\
    \nonumber
    &
    u_{D,d}-\lambda_{P,i:d\in\mathcal{D}_i} -\overline{\eta}_{D,d}+\underline{\eta}_{D,d}+\lambda_{D,d}=0,\\
    &\qquad\qquad\qquad\qquad\qquad \forall d\in \mathcal{D},~~:~P_{D,d}\label{s2}\\
    &
    \sum_{j:(i,j)\in \mathcal{L}}\beta_{ij}(\lambda_{P,j}-\lambda_{P,i}-\overline{\eta}_{L,ij}+\underline{\eta}_{L,ij})=0,\notag\\
    &\qquad\qquad\qquad\qquad\qquad \forall i \in \mathcal{N}/\rm{ref}, ~~ :~\theta_i\label{s3}\\
    \nonumber
    &-\lambda_{G,g}-\lambda_{D,d}-\lambda_{E,d}e_{G,g}\le0,  \\
    &\qquad\qquad\qquad\qquad\qquad \forall g\in\mathcal{G}, d\in \mathcal{D}, ~~ :~\pi_{g,d}\label{s5}\\
    &-c_{D,d}+\lambda_{E,d}=0,~~\!\quad \forall d\in \mathcal{D},~~ :~E_{D,d}\label{s6}\\
    &\overline{\eta}_{L,ij}\ge 0,
    \quad \underline{\eta}_{L,ij}\ge 0, \ \forall(i,j)\in \mathcal{L},\label{d1}\\
    &\overline{\eta}_{G,g}\ge 0,
    \quad \underline{\eta}_{G,g}\ge 0,~~\!\forall g\in \mathcal{G},\label{d2}\\
    &\overline{\eta}_{D,d}\ge 0,
    \quad\underline{\eta}_{D,d}\ge 0, ~~\!\forall d\in \mathcal{D}.\label{d3}
    \end{align}
\end{subequations}
Note that the variables in this problem are $\lambda, \eta$, while the primal variables $P_{G,g}, P_{D,d}, \theta_i, \pi_{g,d}, E_{D,d}$ from Problem (\ref{eq1coo}) are the Lagrange multipliers. Besides, $\lambda_{P,i:g\in\mathcal{G}_i}$ (or $\lambda_{P,i:d\in\mathcal{D}_i}$) represents the dual variable value of constraint (\ref{eq1coa}) on the bus $i$, to which the generator $g$ (or the consumer $d$) is connected.


\subsubsection{KKT conditions} Given the primal and dual problems (\ref{eq1coo}), (\ref{dualp}), we can state the Karush-Kuhn–Tucker (KKT) optimality conditions for our problem. 

\textit{Primal feasibility}: The optimal solutions must satisfy all constraints in the primal problem, i.e., (\ref{eq1coa})-(\ref{eq1cooi}).

\textit{Dual feasibility}: Similarly, the optimal dual solutions must satisfy all constraints in the dual problem, i.e., (\ref{s1})-(\ref{d3}).

\textit{Complementary slackness}: The complementary slackness conditions for the inequality constraints are given by
\begin{subequations}
\label{eqcs}
    \begin{align}
        &\overline{\eta}_{L,ij}\cdot\left(F_{ij}^{\rm{lim}}-\beta_{ij}(\theta_i-\theta_j)\right)=0,  &&\!\!\!\forall(i,j)\in \mathcal{L},\label{c1}\\
        &\underline{\eta}_{L,ij}\cdot\left(F_{ij}^{\rm{lim}}+\beta_{ij}(\theta_i-\theta_j)\right)=0, &&\!\!\!\forall (i,j)\in \mathcal{L},\label{c2}\\
        &\overline{\eta}_{G,g}\cdot(P_{G,g}^{\rm{max}}-P_{G,g})=0,
        &&\!\!\!\forall g\in\mathcal{G}, \label{c7}\\
        &\underline{\eta}_{G,g}\cdot(P_{G,g}-P_{G,g}^{\rm{min}})=0,
        &&\!\!\!\forall g\in\mathcal{G},\label{c3}\\
        &\overline{\eta}_{D,d}\cdot(P_{D,d}^{\rm{max}}-P_{D,d})=0,
        &&\!\!\!\forall d \in \mathcal{D}, \label{c6}\\
        &\underline{\eta}_{D,d}\cdot(P_{D,d}-P_{D,d}^{\rm{min}})=0,
        &&\!\!\!\forall d \in \mathcal{D},\label{c4}\\
        & \pi_{g,d}\cdot\left(-\lambda_{G,g}-\lambda_{D,d}-\lambda_{E,d}e_{G,g}\right)=0, &&\!\!\!\forall  g\!\in\!\mathcal{G}, d\!\in\! \mathcal{D}\label{c5}.
    \end{align}
\end{subequations}

\subsection{Equilibrium Problem}

We derive an equilibrium formulation for problem \eqref{eq1coo} by assigning optimality conditions of problem \eqref{eq1coo}, i.e. \eqref{eq1coa}-\eqref{eq1cooi}, \eqref{s1}-\eqref{d3}, and \eqref{c1}-\eqref{c5}, to different market actors as described below. For each market actor, we then derive a corresponding primal problem. Since we used the optimality conditions of problem \eqref{eq1coo} to derive the equilibrium formulation, the two formulations of problem \eqref{eq1coo} are equivalent and thus have the same optimal solutions. 

\textit{Generators}: We first define the profit maximization problem for an individual generator $g\in\mathcal{G}$ whose optimality conditions are given by (\ref{eq1oc}), (\ref{s1}), (\ref{d2}), (\ref{c7})-(\ref{c3}).
We define the generator output $P_{G,g}$ as the primal variable, while $\overline{\eta}_{G,g}, \underline{\eta}_{G,g}$ are the dual variables and $\lambda_{P,i:g\in\mathcal{G}_i},~\lambda_{G,g}$ are input parameters arising from the price setter and carbon manager problems (described below). 
This leads to the following primal optimization problem for each generator $g\in\mathcal{G}$,
\begin{subequations}
    \label{eq22-2}
    \begin{align}
    \max_{P_{G,g}} ~~ &(\lambda_{P,i:g\in\mathcal{G}_i}+\lambda_{G,g}-c_{G,g})\cdot P_{G,g} \label{eq22a}\\
    \text{s.t.~~}  &P_{G,g}^{\min}\leq P_{G,g}\leq P_{G,g}^{\max}.\label{eq22b}
\end{align}
\end{subequations}
This problem is similar to the profit maximization problem for generators in the standard markets, except the objective \eqref{eq22a} includes the dual variable $\lambda_{G,g}$ associated with the carbon allocation constraint for generators (\ref{eq1cof}).

\textit{Consumers}: Each consumer $d\in \mathcal{D}$ aims to maximize their (carbon-dependent) utility by solving a problem whose optimality conditions are given by \eqref{eq1od}, \eqref{s2}, \eqref{d3}, \eqref{c6}-\eqref{c4}
We define $P_{D,d}$ as the primal variable and $\overline{\eta}_{D,d}, \underline{\eta}_{D,d}$ as dual variables, while $\lambda_{P,i:d\in\mathcal{D}_i},~\lambda_{D,d}$ are input parameters arising from the price setter and carbon manager problems. The primal utility maximization problem can then be defined as follows:
 \begin{subequations}
     \label{eq33-2}
     \begin{align}
    \max_{P_{D,d}} ~~&(u_{D,d}-(\lambda_{P,i: d\in \mathcal{D}_i}-\lambda_{D,d}))\cdot P_{D,d} \label{eq33a1}\\
    \text{s.t.~~} &P_{D,d}^{\min}\leq P_{D,d}\leq P_{D,d}^{\max}.\label{eq33a}
\end{align}
 \end{subequations}
This problem is similar to the utility maximization problem of consumers in standard markets, except the objective \eqref{eq33a1} also includes $\lambda_{D,d}$, the dual variable corresponding to the carbon allocation constraint for consumers (\ref{eq1cog}).

\textit{Transmission owner}: 
Transmission owners
maximize their profit by buying power at one bus and selling it back at another. The optimization problem for the transmission owner is thus based on the primal and dual constraints associated with the transmission limits and power flow constraints, with optimality conditions  (\ref{eq1cob})-(\ref{eq1coblim}), (\ref{eq1coe}), (\ref{s3}), (\ref{d1}), (\ref{c1})-(\ref{c2}). We define $\theta$ as the primal variable and $\overline{\eta}_{L}, \underline{\eta}_{L}$ as dual variables, while $\lambda_{P}$ are the input variables arising from the price setter problem. These optimality conditions give rise to the following problem: 
\begin{align}
    \max_{\theta}\ ~~&\sum_{(i,j)\in \mathcal{L}}(\lambda_{P,j}-\lambda_{P,i})\beta_{ij}(\theta_i-\theta_j)\label{eqtr-2}\\
    \text{s.t.~~}  & {\rm Constraints}\text{ } \eqref{eq1cob}-\eqref{eq1coblim}, \text{ } \eqref{eq1coe}.\notag
\end{align}

\textit{Price setter}: The ISO solves the price setter problem to enforce the nodal power balance constraint for each bus $i\in \mathcal{N}$. This problem can be represented as the complementarity constraint (\ref{eq1coa}) (corresponding to the optimality condition of the price setter) as follows:
\begin{align}
\label{eqmc-2}
    & \sum_{d\in \mathcal{D}_i}P_{D,d}+\sum_{j:(i,j)\in \mathcal{L}}\beta_{ij}(\theta_i-\theta_j)=\sum_{g\in \mathcal{G}_i}P_{G,g}, \quad :\lambda_{P,i}.
\end{align}
In this problem, we define $\lambda_{P,i}$ as a variable, while $P_{G,g}$, $P_{D,d}$, and $\theta$ are input parameters arising from the generator, consumer, and the transmission owner problems, respectively. 

\textit{Carbon manager}:
The carbon manager aims to minimize the total carbon cost by optimally allocating carbon emissions from generators to consumers. 
{\blue Note that the 
``carbon manager'' can be understood as a conceptual actor and does not necessarily represent a new institutional entity. However, the task of the carbon manager, which is to assign a share of electric power and associated emissions from each generator to each load, would have to be implemented as part of the electricity market clearing. In practice, this could, for example, be treated as part of the ISO responsibilities.}
The optimality conditions for the carbon allocation are given by (\ref{eq1cof})-(\ref{eq1cooi}), (\ref{s5})-(\ref{s6}), (\ref{c5}).
We define $\pi_{g,d}$ and $E_{D,d}$ as the primal variables, and $\lambda_{G,g}$, $\lambda_{D,d}$, and $\lambda_{E,d}$ as the dual variables, while $P_{G,g}$ and $P_{D,d}$ are input parameters arising from the generator and consumer problems. This results in the following carbon allocation primal problem:
\begin{subequations}
      \label{eqadp-2}
      \begin{align}
    \max_{\pi, E_D} ~~&-c_D^\intercal E_D&&\!\!\!\!\!\!\!\!\!\!\!\label{eqcarbmanobj}\\
    \text{s.t.~~} & \sum_{d\in \mathcal{D}}\pi_{g,d} = P_{G,g},\qquad &&\!\!\!\!\!\!\!\!\!\!\!\forall g \in \mathcal{G}, \quad :~\lambda_{G,g}\label{eqadpa}\\
    & \sum_{g\in\mathcal{G}}\pi_{g,d} = P_{D,d},\qquad && \!\!\!\!\!\!\!\!\!\!\!\forall d \in \mathcal{D}, \quad :~\lambda_{D,d}\label{eqadpb}\\
    & \sum_{g\in\mathcal{G}}e_{G,g}\pi_{g,d} = E_{D,d}, \qquad &&\!\!\!\!\!\!\!\!\!\!\! \forall d \in \mathcal{D},\label{eqadpc}\\
    & \pi_{g,d}\geq 0,\qquad &&\!\!\!\!\!\!\!\!\!\!\! \forall g \in \mathcal{G}, \ \forall d \in \mathcal{D}.\label{eqadpd}  
\end{align}
\end{subequations}
The objective function \eqref{eqcarbmanobj} minimizes the cost of carbon allocation to consumers. The dual variables $\lambda_{G,g}$ and $\lambda_{D,d}$ for constraints (\ref{eqadpa}) and (\ref{eqadpb}), become inputs to the generator and consumer optimization problems, respectively.

An important conclusion of the carbon manager problem is that for a given generation and load dispatch $P_G, P_D$, it is cost optimal to assign the lowest emitting generation to the consumers with the highest carbon cost.


\section{Carbon-Adjusted Prices}
\label{sec4}
We next analyze how the integration of consumer-based carbon costs and a carbon allocation mechanism impact prices for consumers and generators. 
\subsection{Carbon-Adjusted Prices}
From the equilibrium model, we observe that generators are paid $\lambda_{P,i:g\in \mathcal{G}_i}+\lambda_{G,g}$ and loads pay $\lambda_{P,i:d\in \mathcal{D}_i}-\lambda_{D,d}$ for their electricity. Thus, we propose to define carbon-adjusted electricity prices as follows.
\begin{definition}[Carbon-Adjusted Prices]
    The carbon-adjusted prices are defined as
    \begin{align}
        &\lambda_{P,i:g\in \mathcal{G}_i}+\lambda_{G,g} \quad\text{for generators $g\in\mathcal{G}$,} \notag \\
        &\lambda_{P,i:d\in \mathcal{D}_i}-\lambda_{D,d} \quad\text{for consumers $d\in\mathcal{D}$.} \notag 
    \end{align}
\end{definition}
Note that the carbon-adjusted prices may be different for generators and consumers who are at the same bus $i$, as the carbon adjustments $\lambda_{G,g},~\lambda_{D,d}$ may differ even when $\lambda_{P,i}$ is the same for both. 

It is tempting to interpret $\lambda_{P}$ directly as the nodal electricity price and assume that the variables $\lambda_{G},~\lambda_{D}$ represent the value of the carbon emissions.
However, the exact values of $\lambda_{P},~\lambda_{G}$ and $\lambda_{D}$ are non-unique,
as shifting them by a $d\lambda\in \mathbb{R}$ gives rise to a new set of optimal dual variables,
\begin{subequations}
\label{eqms}
    \begin{align}&\tilde{\lambda}_{P,i}=\lambda_{P,i}+d\lambda,\label{eq9a}\\
    &\tilde{\lambda}_{G,g} = \lambda_{G,g} - d\lambda,\\
    &\tilde{\lambda}_{D,d} = \lambda_{D,d} + d\lambda,
    \end{align}
\end{subequations}
This is because the combinations of $\lambda_{P},~\lambda_{G,g}$ and $\lambda_{D,d}$ in the dual problem \eqref{dualp}, and in the objectives of problems \eqref{eq22-2}, \eqref{eq33-2}, and \eqref{eqtr-2} lead to cancellations of $d\lambda$. 
Since the optimal values of the electricity price and carbon adjustments are non-unique, we analyze their relative size rather than their absolute values.


\subsection{Ordering of Carbon-Adjustments}
Considering {\blue consumer-based} carbon costs may result in a generation dispatch that is different from the standard market clearing, leading to higher generation costs and increased prices to consumers. It is therefore important to understand who will pay for this increased cost. While variations in $\lambda_P$ remain due to transmission congestion, we next prove how the consumer-based carbon cost $c_D$ and generation emission factors $e_G$ impact the carbon-adjustments $\lambda_G, \lambda_D$, and thus the total cost of electricity.
Our main result, stated in Theorem \ref{theoremodv}, shows that the generators with lower carbon emissions $e_G$ will receive a higher carbon adjustment and be paid comparatively more for their generation, whereas consumers with a higher carbon cost $c_D$ will receive a lower carbon adjustment and pay comparatively more for their (lower carbon) electricity. 
This suggests that our proposed market clearing is fair in the sense that the most emitting generators are penalized with lower payments, and that the consumers that submit higher carbon costs contribute more  to cover the increases in generation cost that arise from prioritizing low carbon generation.
\begin{theorem}[Ordering of Carbon-Adjustments]
\label{theoremodv}
For a set of generators $\mathcal{G}$ with increasing emission factors $e_{G,(1)} \leq e_{G,(2)} \leq \cdots \leq e_{G,(|\mathcal{G}|)}$, the corresponding generator carbon-adjustments will be decreasing,
$$ \lambda_{G,(1)}\geq\lambda_{G,(2)}\geq\cdots \geq\lambda_{G,(|\mathcal{G}|)}. $$
For a set of consumers $\mathcal{D}$ with decreasing carbon-costs $c_{D,(1)}\geq c_{D,(2)}\geq \cdots \geq c_{D,(|\mathcal{D}|)}$, 
the corresponding consumer carbon-adjustments will be increasing, 
$$ \lambda_{D,(1)}\leq\lambda_{D,(2)}\leq \cdots \leq \lambda_{D,(|\mathcal{D}|)} $$
\end{theorem}

\begin{proof}
The dual constraints \eqref{s5}, \eqref{s6} require that
\begin{equation}
    \lambda_{G,g}+\lambda_{D,d} +c_{D,d}e_{G,g} \geq 0,\label{ineq-cond}
\end{equation} 
for all generator-consumer pairs $(g,d)~\in\mathcal{G}\times\mathcal{D}$.
For any generator-consumer pair with a non-zero power allocation $\pi_{g,d}>0$, the complementary slackness condition \eqref{c5} requires
\begin{equation}\lambda_{G,g}+\lambda_{D,d} +c_{D,d}e_{G,g} = 0,\label{eq-cond}
\end{equation} 
Further, from the constraints \eqref{eq1cof}, \eqref{eq1cog}, we know that all generators $g$ with $P_{G,g}>0$ must have $\pi_{g,q}>0$ for at least one load $q\in\mathcal{D}$. Thus, \eqref{eq-cond} holds for at least one load $q\in\mathcal{D}$. A similar argument can be made for all loads with $P_{D,d}>0$.

We next prove the ordering of the generator carbon adjustments.
Consider two generators $h$ and $k$ with carbon emission factors $e_{G,h} \leq e_{G,k}$, and a consumer $\ell$ that is served by generator $k$, i.e. $\pi_{k,\ell}>0$. We then have that
\begin{align}
    &\lambda_{G,k}+\lambda_{D,\ell} = -c_{D,\ell}e_{G,k} \label{eq:gen-cond1},\\
    &\lambda_{G,h}+\lambda_{D,\ell} \ge -c_{D,\ell}e_{G,h} .\label{eq:gen-cond2}
\end{align}
By subtracting \eqref{eq:gen-cond1} from \eqref{eq:gen-cond2}, we obtain the expression 
\begin{equation}
    \lambda_{G,h}-\lambda_{G,k} \geq c_{D,\ell}(e_{G,k}- e_{G,h})\geq 0 \label{eq:gen-cond-result},
\end{equation}
where the last inequality follows from the fact that $c_{D,\ell}\geq 0$ and 
$e_{G,h} \leq e_{G,k}$. From \eqref{eq:gen-cond-result}, we can thus conclude that all generators $k$ with an emissions factor $e_{G,k} \geq e_{G,h}$  will have a smaller carbon adjustment $\lambda_{G,k}\leq\lambda_{G,h}$ compared to generator $h$. Repeating this analysis for all generators $g$ in order of increasing emissions factor $e_{G,g}$, we can show that the generator with the smallest emissions factor $e_{G,g}$ will have the highest $\lambda_{G,g}$, the generator with the second smallest emissions factor $e_{G,g}$ will have the second highest $\lambda_{G,g}$, and so on. Thus, a set of generators $\mathcal{G}$ with \emph{increasing} emission factors $e_{G,(1)} \leq e_{G,(2)} \leq \cdots \leq e_{G,(|\mathcal{G}|)}$ will have decreasing carbon adjustments
$$ \lambda_{G,(1)}\geq\lambda_{G,(2)}\geq\cdots \geq\lambda_{G,(|\mathcal{G}|)}.$$

Using similar arguments, we can prove that  
a set of consumers $\mathcal{D}$ with \emph{decreasing} $c_{D,(1)}\geq c_{D,(2)}\geq \cdots \geq c_{D,(|\mathcal{D}|)}$ will have increasing carbon adjustments.
$$ \lambda_{D,(1)}\leq\lambda_{D,(2)}\leq \cdots \leq \lambda_{D,(|\mathcal{D}|)}. \qedhere$$ 
\end{proof}

Using the relationship \eqref{eq:gen-cond1} for two generators serving the same load or two loads are served by the same generator, we can derive specific differences in carbon adjustments.
\begin{corollary}
\label{lemma555}
Let $g_1, g_2$ be two different generators with carbon intensities $e_{G,g_1}, e_{G,g_2}$ serving consumer $l$, i.e. $\pi_{g_1,l}>0$ and $\pi_{g_2,l}>0$.  
The difference in their carbon adjustments $\lambda_{G,g_1}, \lambda_{G,g_2}$ is given by:
$$
    \lambda_{G,g_2}-\lambda_{G,g_1}=c_{D,l}(e_{G,g_1}-e_{G,g_2}).
$$
\end{corollary}
\begin{corollary}
\label{lemma3}
Let $d_1, d_2$ be two different consumers with carbon costs $c_{D,d_1}, c_{D,d_2}$ who are served by the same generator $r$, i.e. $\pi_{r,d_1}>0$ and $\pi_{r,d_2}>0$. 
The difference in their carbon adjustments $\lambda_{D,d_1}, \lambda_{D,d_2}$ is given by:
$$
    \lambda_{D,d_2}-\lambda_{D,d_1}=e_{G,r}(c_{D,d_1}-c_{D,d_2}).
$$
\end{corollary}

These corollaries have two interesting implications. First, loads with the same carbon cost $c_D$ (or generators with the same emissions intensity $e_G$) will have the same carbon adjustment $\lambda_D$ (or $\lambda_G$). Second, generators who serve a load with zero carbon cost $c_D=0$ (or loads served by generators with zero emissions $e_G=0$) will have the same $\lambda_D$ (or $\lambda_G$).


\section{Market-Clearing Properties}
\label{secmp}
\textcolor{black}{We consider} four desirable properties of market-clearing mechanisms: market efficiency, incentive compatibility, revenue adequacy, and individual rationality. 
\textcolor{black}{Leveraging} the equilibrium model, we show that our proposed model satisfies properties similar to standard electricity markets based on locational marginal prices (LMPs), as discussed below.

\subsection{Revenue Adequacy}
\textcolor{black}{We first prove that} our market clearing mechanism is revenue adequate, that is 
the payment to the ISO from the consumers is always higher than or equal to their total payments to generators, the transmission owner, and the carbon manager. 
\begin{proposition}
\label{prop}
Our model satisfies revenue adequacy.
\end{proposition}
\begin{proof}
To prove revenue adequacy, we need to show that, at optimum, the following inequality holds,
    \begin{align}
        &\sum_{d\in \mathcal{D}}(\lambda_{P,i:d\in \mathcal{D}_i}-\lambda_{D,d})P_{D,d}
        -\sum_{g\in \mathcal{G}}(\lambda_{P,i:g\in \mathcal{G}_i}+\lambda_{G,g})P_{G,g}
        \notag\\
        &\qquad-\!\sum_{i\in \mathcal{N}}\lambda_{P,i}\!\!\!\! 
        \sum_{j, (i,j)\in \mathcal{L}}\!\!\!\beta_{ij}(\theta_j-\theta_i)-\sum_{d\in \mathcal{D}}c_{D,d}E_{D,d}
        \!\geq\! 0.\label{eq20}
    \end{align}
To achieve this, we first multiply $\lambda_{P,i}$ on both sides of the power balance constraint (\ref{eq1coa}), sum across all buses and rearrange terms to obtain
    \begin{align}
        &\sum_{d\in \mathcal{D}}\lambda_{P,i:d\in \mathcal{D}_i}P_{D,d}-\sum_{g\in \mathcal{G}}\lambda_{P,i:g\in \mathcal{G}_i}P_{G,g} \notag\\
        &\qquad+\sum_{i\in \mathcal{N}}\lambda_{P,i}\sum_{j:(i,j)\in \mathcal{L}}\beta_{ij}(\theta_i-\theta_j)=0.\label{eq22u}
    \end{align}
By subtracting (\ref{eq22u}) from (\ref{eq20}) and rearranging the remaining terms, we get the following condition for revenue adequacy:
\begin{equation}
\sum_{d\in \mathcal{D}}\lambda_{D,d}P_{D,d}+\sum_{g\in \mathcal{G}}\lambda_{G,g}P_{G,g}+\sum_{d\in \mathcal{D}}c_{D,d}E_{D,d}\leq 0.\label{eq23}
\end{equation}
Using (\ref{eq1cof}), (\ref{eq1cog}) to express $P_{D,d},~P_{G,g}$ and $E_{D,d}$ in terms of $\pi$, the left-hand side of \eqref{eq23} becomes
\begin{align}
\nonumber
&\sum_{d\in \mathcal{D}}\lambda_{D,d}\!\sum_{g\in \mathcal{G}}\!\pi_{g,d}+\!\!\sum_{g\in \mathcal{G}}\lambda_{G,g}\!\sum_{d\in \mathcal{D}}\!\pi_{g,d}+\sum_{d\in \mathcal{D}}c_{D,d}\sum_{g\in \mathcal{G}}e_{G,g}\pi_{G,d}
\\
&\qquad=\sum_{g\in \mathcal{G}}\sum_{d\in \mathcal{D}}\pi_{g,d}(\lambda_{D,d}+\lambda_{G,g}+e_{G,g}c_{D,d}) = 0.
\end{align}
The last equality arises from the fact that $\pi_{g,d}\ge0$ and, for $\pi_{g,d}>0$, the complementary slackness condition \eqref{c5} requires that
    $\lambda_{G,g}+\lambda_{D,d} + c_{D,d}e_{G,g} = 0\label{equ26}$.
This shows that the proposed market-clearing mechanism is revenue adequate, and, in fact, budget balanced since \eqref{eq20} will be satisfied with equality. 
\end{proof}
\subsection{Individual Rationality}
The individual rationality property \textcolor{black}{implies} that generators and consumers \textcolor{black}{have an incentive to participate in the market (rather than opting out), which requires that they} always recover their operational costs and do not incur a loss. We prove \textcolor{black}{that} our model \textcolor{black}{satisfies cost recovery for the case where $P_{G}^{\min}=P_{D}^{\min}=0$}.
\begin{proposition}
\label{prop}
Our model satisfies individual rationality given $P_{G}^{\min}=P_{D}^{\min}=0$.
\end{proposition}
\begin{proof}
\textcolor{black}{Cost recovery for individual participants} is ensured when we can guarantee a positive objective function value for all generators and consumers,
\begin{align}
    (\lambda_{P,i:g\in \mathcal{G}_i}+\lambda_{G,g}-c_{G,g})P_{G,g}&\ge 0, \label{eq:gen_cost_recovery}\\
    (u_{D,d}-\lambda_{P,i:d\in \mathcal{D}_i}+\lambda_{D,d})P_{D,d}&\ge 0. \label{eq:load_cost_recovery}
\end{align}
If $P_{G}^{\min}=P_{D}^{\min}=0$ for all generators and consumers, we can always set $P_{G,g}=P_{D,d}=0$ to satisfy \eqref{eq:gen_cost_recovery}, \eqref{eq:load_cost_recovery}.
\end{proof}
Note that if $P_{G}^{\min}>0$ or $P_{D}^{\min}>0$, individual rationality is not guaranteed (as is the case in current electricity markets). 

\subsection{Market Efficiency and Incentive Compatibility}
Since the optimality conditions of the centralized model and the equilibrium model are the same, the two problems are equivalent. 
This implies that if the dual variables  $\lambda_{G},~\lambda_{D}$, and $\lambda_{P}$ are used to define prices, 
the solution to the centralized model
aligns with the solution to the individual problems solved by each market participant. This suggests that our proposed model is efficient {\blue under the assumption of a competitive market (without strategic bidding or market power)}, as no player has an incentive to unilaterally deviate from the socially optimal outcome. 

An incentive-compatible market incentivizes players to \textcolor{black}{truthfully} bid their marginal costs \textcolor{black}{even in the presence of market power. However, according to Hurwicz impossibility theorem \cite{hurwicz1972, myerson1983efficient, BANERJEE1994397}, no market mechanism where individuals directly report their private information (e.g. their costs) and the market mechanism's rules determine the outcome can simultaneously ensure efficiency, individual rationality, cost recovery, and incentive compatibility. Since our market design is efficient, individually rational, and ensures cost recovery, it follows that it cannot be incentive compatible.} Specifically, players can exercise market power and increase their profit by strategically adapting their costs (generator costs, consumer bids for power or carbon costs) or the quantities (generation or consumption limits) 
offered to the market. This lack of incentive compatibility suggests that the \textcolor{black}{overall} market is not
inherently efficient and that market power mitigation is needed, as is the case in current electricity markets. {\blue A specific counterexample of incentive compatibility in existing electricity markets (without carbon costs) can be found in Example 3 of Section IV in \cite{tang2013nash} or the two-bus example of Section VI.A in \cite{xu2015efficient}.} 

 \section{Special Versions of Carbon Cost Model}
 \label{secmg}
We next show that the proposed carbon cost model generalizes two commonly implemented market models.


\begin{proposition}[Equivalence to standard market clearing]
    If $c_{D,d}=0$ for all $d\in \mathcal{D}$, our model is equivalent to the standard (i.e. carbon agnostic) market clearing.
\end{proposition}

\begin{proof}
In the standard carbon-agnostic model, consumers submit no information about their carbon costs, i.e., all $c_D = 0$. 
If we set all consumer carbon costs $c_D = 0$ in our model, the carbon cost term $c_D E_D$ in the objective trivially becomes zero. Further, since there will always exist a feasible allocation of generation to load satisfying \eqref{eq1cof}-\eqref{eq1cooi}, which is also optimal, we can omit the carbon allocation constraints \eqref{eq1cof}-\eqref{eq1cooi} from our problem, leading to the following model,
\begin{align}
    \label{eqcon-2}
        \max_{P_G, P_D, \theta}~~&u_{D}^\intercal P_{D}-c_{G}^\intercal P_{G}\\
    \text{s.t.~~} & \rm{Constraints} \ (\ref{eq1coa})-(\ref{eq1coe}),\notag
\end{align}
which is the standard carbon-agnostic model.
\end{proof}

\begin{proposition}[Equivalence to carbon tax on generation]
 If all consumers $d\in \mathcal{D}$ have the same carbon cost $c_{D,d}=c_{tax}$, our model is equivalent to introducing the carbon tax on generators with a tax rate $c_{tax}$.   
\end{proposition}

\begin{proof}
Introduce a carbon tax $c_{tax}$ on generators in the standard model \eqref{eqcon-2} leads to a change in generation cost, i.e.,
    \begin{align}
    \label{eqcon-5}
\max_{P_G, P_D, \theta}~~&\ u_D^\intercal P_D-(c_{tax} e_G + c_G)^\intercal P_G\\
    \text{s.t.~~} & \rm{Constraints} \ (\ref{eq1coa})-(\ref{eq1coe})\notag.
\end{align}


The consumer-based carbon cost model (\ref{eq1coo}) with all $c_D=c_{tax}$ is given by
\begin{align}
\max_{P_G, P_D, \theta, \pi, E_D}~~&\ u_D^\intercal P_D-c_{tax}\mathbf{1}^\intercal E_D -c_G^\intercal P_G\label{eqcon-3}\\
    \text{s.t.~~} & \rm{Constraints} \ (\ref{eq1coa})-(\ref{eq1cooi})\notag,
\end{align}
where $\mathbf{1}\in \mathbb{R}^{|\mathcal{D}|}$ is a vector with all elements equal to $1$. 
We next use constraints (\ref{eq1cho}) and (\ref{eq1cof}) to substitute $E_D$, i.e.
\begin{equation}
    \mathbf{1}^\intercal E_D \!=\!\! 
    \sum_{d\in\mathcal{D}} \!E_{D,d} 
    \!=\!\! \sum_{d\in\mathcal{D}}\sum_{g\in\mathcal{G}} e_{G,g}\pi_{g,d}
    \!=\!\! \sum_{g\in\mathcal{G}} e_{G,g}P_{G,g}
    \!=
    e_{G}^\intercal P_{G}.
\end{equation}
Considering this substitution, we no longer need the variables $E_D$ and $\pi$ or the constraints (\ref{eq1cog}) and (\ref{eq1cooi}) to define them.
We can therefore restate our problem as  
\begin{align}
    \notag
\max_{P_G, P_D, \theta}~~&\ u_D^\intercal P_D-c_{tax} e_G^\intercal P_G -c_G^\intercal P_G\\
    \text{s.t.~~} & \rm{Constraints} \ (\ref{eq1coa})-(\ref{eq1coe}),\notag
\end{align}
which is the standard market with a carbon tax \eqref{eqcon-5}.
\end{proof}

\section{Case study}
\label{secns}

We next provide a numerical case study to demonstrate our theoretical results and illustrate how the proposed model impacts market clearing results. 
The optimization problem is solved using both GAMs \cite{GAMs} and Julia \cite{bezanson2017julia}. 
\begin{table}[b]
\footnotesize
    \centering
    \renewcommand{\arraystretch}{1.1}
    \setlength{\tabcolsep}{4pt}
    \caption{\small Three bus system parameters.}
    \label{tab: 3-bus-params}
    \begin{tabular}{c|ccc|cccc}
        \hline
        \multirow{2}{*}{\textbf{Bus}} & \multicolumn{3}{c|}{\textbf{Consumers}} & \multicolumn{4}{c}{\textbf{Generators}} \\ \cline{2-8}
        & $\boldsymbol{P_{d}^{\rm{min}}}$ & $\boldsymbol{P_{d}^{\rm{max}}}$ & \textbf{$\boldsymbol{u_D}$} & $\boldsymbol{P_{g}^{\rm{min}}}$ & $\boldsymbol{P_{g}^{\rm{max}}}$ & \textbf{$\boldsymbol{c_G}$} &  \textbf{$\boldsymbol{e_G}$} \\
        \hline
        \textbf{1} & 0 & 15 & 18 & 0 & 20 & 8 & 0.6 \\
        \textbf{2} & 0 & 15 & 18 & 0 & 10 & 10 & 0.2 \\
        \textbf{3} & 0 & 15 & 18 & 0 & 25 & 6 & 1 \\
        \hline
    \end{tabular}
\end{table}
\vspace{-0.3cm}


\subsection{Simplified Three-bus Illustration Example}
We first consider a simplified three-bus system with one generator and one consumer connected at each bus, adapted from Example 6.2.2 in \cite{gabriel2012complementarity}.
We define carbon emission factors $e_G$ such that generator 3, the cheapest generator, is also the most emission-intense. To highlight the impact of consumer carbon costs on results, we harmonize the load parameters to $P_{D}^{\min}=0$MW, $P_{D}^{\max}=15$MW and $u_D=\$18$/MWh for all loads. The system parameters are summarized in Table \ref{tab: 3-bus-params}. 
{\blue Note that, to more clearly highlight the impacts of introducing consumer-based carbon costs and carbon allocation in the market model, we consider an uncongested system (without active transmission constraints) in this example. As a result, the electricity price $\lambda_P$ will be uniform throughout the system. 
For results highlighting the impact of congestion, the reader is referred to the case study involving the RTS-GMLC system below.
}

\begin{table*}[t]
\footnotesize
\centering
\caption{\small Impact of Carbon Costs on Power Dispatch and Emissions. Bolded carbon-adjusted prices indicate prices that are equal to the generator cost or demand utility.}
\label{tab:carbon_pricing_results}

\begin{tabular}{>{\centering\arraybackslash}p{1.5cm}|
                >{\centering\arraybackslash}p{0.5cm} >{\centering\arraybackslash}p{0.5cm} >{\centering\arraybackslash}p{0.5cm}|
                >{\centering\arraybackslash}p{0.5cm} >{\centering\arraybackslash}p{0.5cm} >{\centering\arraybackslash}p{0.5cm}|
                >{\centering\arraybackslash}p{0.5cm} >{\centering\arraybackslash}p{0.5cm} >{\centering\arraybackslash}p{0.5cm}|
                >{\centering\arraybackslash}p{1.2cm}| >{\centering\arraybackslash}p{1.2cm}|
                >{\centering\arraybackslash}p{0.5cm} >{\centering\arraybackslash}p{0.5cm} >{\centering\arraybackslash}p{0.5cm}|
                >{\centering\arraybackslash}p{0.5cm} >{\centering\arraybackslash}p{0.5cm} >{\centering\arraybackslash}p{0.5cm}|
                >{\centering\arraybackslash}p{1cm}| >{\centering\arraybackslash}p{1cm}|
                >{\centering\arraybackslash}p{0.5cm} >{\centering\arraybackslash}p{0.5cm} >{\centering\arraybackslash}p{0.5cm}}
\hline
\multirow{5}{1.5cm}{\centering \textbf{Case definition}} & \multicolumn{3}{>{\centering\arraybackslash}p{1.5cm}|}{\multirow{4}{1.5cm}{\centering \textbf{Carbon cost [\$/CO$_2$]}}} & \multicolumn{3}{>{\centering\arraybackslash}p{1.5cm}|}{\multirow{4}{1.5cm}{\centering \textbf{Load dispatch [MW]}}} & \multicolumn{3}{>{\centering\arraybackslash}p{1.5cm}|}{\multirow{4}{1.5cm}{\centering \textbf{Gen. dispatch [MW]}}} & \multirow{5}{1cm}{\centering \textbf{Total load/ gen. [MWh]}} & \multirow{5}{1cm}{\centering \textbf{Total gen. cost [\$]}} & \multicolumn{3}{>{\centering\arraybackslash}p{1.5cm}|}{\textbf{Gen. carbon-adj. price [\$/MWh]} } & \multicolumn{3}{>{\centering\arraybackslash}p{1.5cm}|}{\textbf{Load carbon-adj. price [\$/MWh]}} & \multirow{5}{1cm}{\centering\textbf{Total E. [tCO$_2$]}} & \multirow{5}{1cm}{\centering \textbf{Avg E. [tCO$_2$ /MWh]}} & \multicolumn{3}{>{\centering\arraybackslash}p{1.5cm}}{\multirow{4}{1.5cm}{\centering \textbf{Emission allocation [tCO$_2$]}}} \\
\cline{2-10} \cline{13-18} \cline{21-23}
& $d_1$ & $d_2$ & $d_3$ & $d_1$ & $d_2$ & $d_3$ & $g_1$ & $g_2$ & $g_3$ & & & $g_1$ & $g_2$ &$g_3$ & $d_1$ & $d_2$ & $d_3$ & & & $d_1$ & $d_2$ & $d_3$ \\
\hline
\textbf{Standard} & 0 & 0 & 0 & 15 & 15 & 15 & 20 & 0 & 25 & 45 & 310 & \multicolumn{6}{>{\centering\arraybackslash}p{3cm}|}{10} & 37 & 0.82 &  & N/A &  \\
\hline
\textbf{Carbon tax} & 15 & 15 & 15 & 15 & 10 & 5 & 20 & 10 & 0 & 30 & 260 & 9 & 15 & 3 & \textbf{18} & \textbf{18} & \textbf{18} & 14 & 0.47 & 9 & 2 & 3 \\
\hline
\multirow[=]{6}{1.5cm}{\centering \textbf{Non-uniform\\carbon\\costs}} 
 & 0 & 15 & 0  & 15 & 15 & 15 & 10 & 10 & 25 & 45 & 330 & \textbf{8} & 14 & 8 & 8 & 17 & 8  & 33 & 0.73 & 13 & 5 & 15 \\
 & 0 & 15 & 5  & 15 & 15 & 15 & 20 & 10 & 15 & 45 & 350 & \textbf{8} & 14 & \textbf{6} & 6 & 17 & 11 & 29 & 0.64 & 15 & 5 & 9 \\
 & 0 & 15 & 10 & 15 & 15 & 15 & 20 & 10 & 15 & 45 & 350 & 9 & 15 & \textbf{6} & 6 & \textbf{18} & 15  & 29 & 0.64 & 15 & 5 & 9 \\
 & 0 & 15 & 15 & 15 & 15 & 15 & 20 & 10 & 15 & 45 & 350 & 9 & 15 & \textbf{6} & 6 & \textbf{18} & \textbf{18} & 29 & 0.64 & 15 & 9 & 5 \\
 & 0 & 15 & 20 & 15 & 15 & 10  & 15 & 10 & 15 & \textcolor{black}{40} & \textcolor{black}{310} & \textbf{8} & 14 & \textbf{6} & 6 & 17 & \textbf{18} & \textcolor{black}{26} & \textcolor{black}{0.65}& 15 & 9 & 2 \\
 & 0 & 15 & 25 & 15 & 15 & 0  & 5  & 10 & 15 & 30 & 230 & \textbf{8} & 14 & \textbf{6} & 6 & 17 & 19 & 20 & 0.67 & 15 & 5 & 0 \\
\hline
\end{tabular}
\end{table*}

\subsubsection{Impact of Carbon Costs on Market Clearing}
We first provide an example of how consumer-based carbon costs impact the market-clearing outcomes. Specifically, we compare market-clearing outcomes with zero carbon costs $c_D = [0,0,0]$ (equivalent to a standard, carbon-agnostic market clearing) to outcomes with uniformly high carbon costs $c_D = [15, 15, 15]$ (i.e. equivalent to adding a unifying carbon tax) and our proposed model with non-uniform carbon costs. 
The results are listed in Table \ref{tab:carbon_pricing_results}.

\begin{figure}[t]
  \centering
  \begin{subfigure}{\linewidth}
    \centering
    \vspace{-0.2cm}
\includegraphics[width=\linewidth]{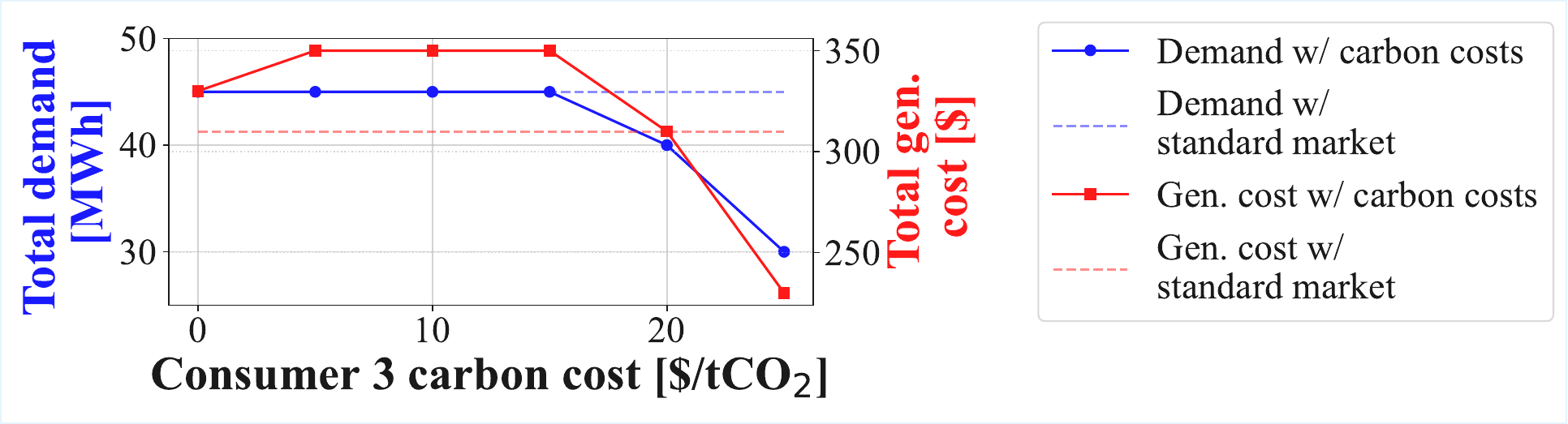}
    \caption{\small Total demand (red) and total generation cost (blue). Dashed lines indicate demand and generation cost from the standard market.} 
    \vspace{+2pt}
    \label{fig:totaldemand}
  \end{subfigure}
   
  \begin{subfigure}{\linewidth}
    \centering
\includegraphics[width=\linewidth]{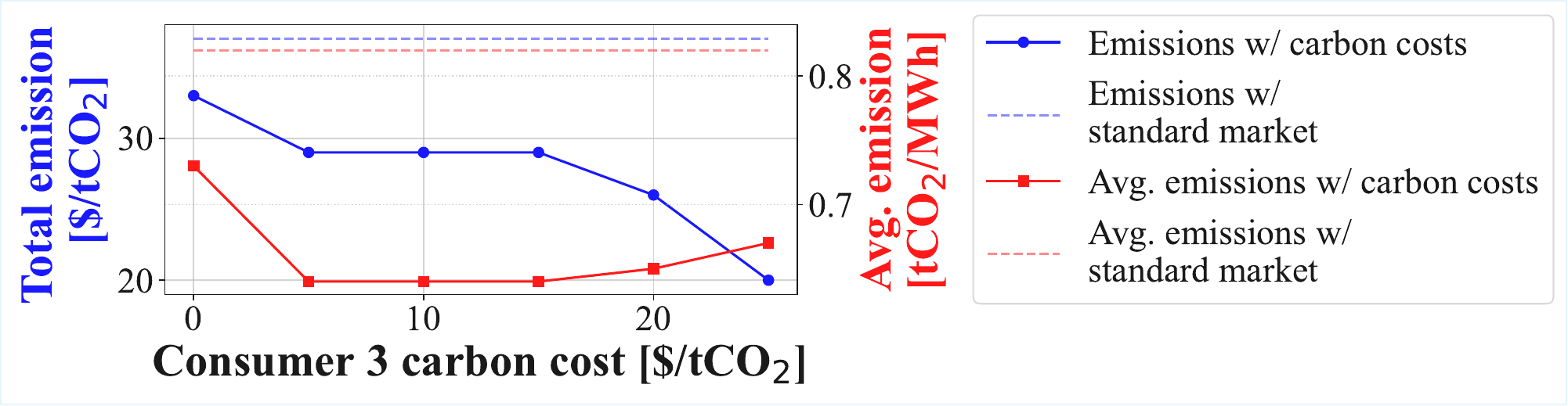}
    \caption{\small Total (blue) and average (red) carbon emissions. Dashed lines show emissions from the standard market clearing.}
    \vspace{+2pt}
    \label{fig:Emissions}
  \end{subfigure}
  
  \begin{subfigure}{\linewidth}
    \centering
\includegraphics[width=\linewidth]{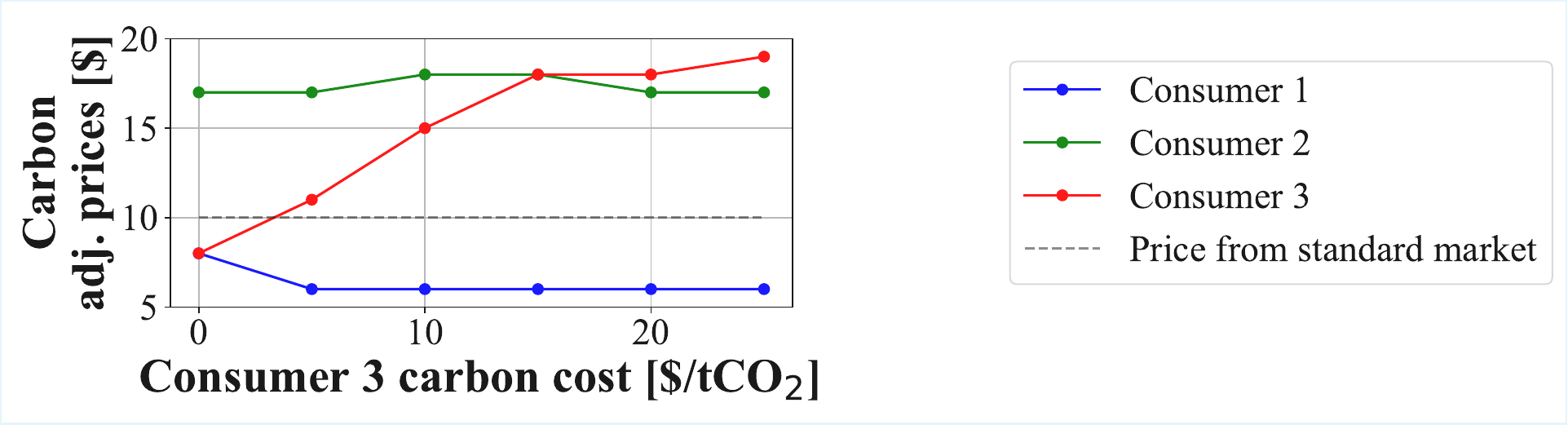}
    \caption{\small Carbon-adjusted prices for Consumer 1 (blue), Consumer 2 (green) and Consumer 3 (red). Black dashed line indicates price from standard market clearing.}
\label{fig:DemandAdjPrices}
  \end{subfigure}

  \caption{\small Impact of increasing carbon costs of Consumer 3 on market clearing results.}
      \vspace{-0.5cm}
  \label{fig:impactc}
\end{figure}

\begin{table*}[htbp]
    \footnotesize
    \centering
    \renewcommand{\arraystretch}{1}
    \setlength{\tabcolsep}{1pt}
    \caption{\small Comparison of Multiple Optimal Solutions for $c_D = \{0,15,20\}$. Values common to all three solutions are only listed for S2.}
    \label{tab:model_comparison}
    \begin{tabular}{>{\centering\arraybackslash}p{2cm}|
                >{\centering\arraybackslash}p{1.75cm} >{\centering\arraybackslash}p{1.75cm} >{\centering\arraybackslash}p{1.75cm}
                >{\centering\arraybackslash}p{1.75cm}| >{\centering\arraybackslash}p{0.5cm} >{\centering\arraybackslash}p{0.5cm}
                >{\centering\arraybackslash}p{0.5cm}| >{\centering\arraybackslash}p{0.5cm} >{\centering\arraybackslash}p{0.5cm}
                >{\centering\arraybackslash}p{0.5cm}|
                >{\centering\arraybackslash}p{0.5cm} >{\centering\arraybackslash}p{0.5cm}
                >{\centering\arraybackslash}p{0.5cm}| >{\centering\arraybackslash}p{0.5cm} >{\centering\arraybackslash}p{0.5cm}
                >{\centering\arraybackslash}p{0.5cm}}
                \hline
                \multirow{3}{2cm}{\centering \textbf{Solutions}} & \multirow{3}{1.75cm}{\centering \textbf{Objective value [\$]}} &  \multirow{3}{1.75cm}{\centering \textbf{Consumer utility [\$]}} & \multirow{3}{1.75cm}{\centering \textbf{Carbon cost [\$]}} &  \multirow{3}{1.75cm}{\centering \textbf{Generation cost[\$]}} & \multicolumn{6}{>{\centering\arraybackslash}p{3.6cm}|}{\textbf{Power consumption $P_D$ and generation $P_G$[MWh]}} & \multicolumn{6}{>{\centering\arraybackslash}p{3cm}}{\multirow{2}{3cm}{\centering\textbf{Carbon-adjusted prices [\$/MWh]}}}\\
                \cline{6-11} \cline{12-17}
                & & & & & $d_1$ & $d_2$ & $d_3$ & $g_1$ & $g_2$ & $g_3$ & $g_1$ & $g_2$ & $g_3$ & $d_1$ & $d_2$ & $d_3$ \\
                \hline
                \textbf{S1} &\multirow{3}{*}{235} & 540& 75& 230 & \multirow{3}{*}{15} & \multirow{3}{*}{15} & 0 & 5 & \multirow{3}{*}{10} & \multirow{3}{*}{15} & \multirow{3}{*}{8} & \multirow{3}{*}{14} & \multirow{3}{*}{6} & \multirow{3}{*}{6} & \multirow{3}{*}{17} & \multirow{3}{*}{18} \\
                \textbf{S2} & &583 & 99& 249 &  &  & 2.4 & 7.4 &  &  &  &  &  &  &  &  \\
                \textbf{S3} & &720&175&310& & &10& 15 & & & & & & & &\\
                \hline
\end{tabular}
\end{table*}

The standard market clearing with carbon costs $c_D = [0,0,0]$ dispatches the cheapest and most polluting generators. All generators and consumers have an electricity price of $\lambda_P = \$10$/MWh, with zero carbon adjustments. 
The total load is 45MWh with a generation cost of \$310, with total emissions of 37 tCO$_2$ and average emissions of 0.82 tCO$_2$/MWh.

With a uniform carbon costs $c_D = [15, 15, 15]$, corresponding to a carbon tax of \$15/tCO$_2$, we observe that 
each generator receives differing carbon-adjusted prices
with the more polluting generators facing lower prices. 
Generator 3, with the highest emissions, produces no power as its carbon-adjusted price is lower than the generation cost.
Generators 1 and 2 produce their maximum amount of power, leading to 30 MWh of generation. 
Due to the uniform carbon price, the carbon-adjusted electricity price of $\lambda_P - \lambda_D= \$18$/MWh is the same for all consumers, and equals the consumer utility $u_D= \$18$/MWh. 
The total generation cost is \$260, while the total and average carbon emissions are reduced to 14 tCO$_2$ and 0.47 tCO$_2$/MWh, respectively. 
Compared to the standard market clearing, the carbon tax reduces emissions both by 
prioritizing lower emitting generators
and by reducing load.


We next investigate the impact of non-uniform carbon costs. 
We fix the carbon costs of consumers 1 and 2 to \$0/tCO$_2$ and \$15/tCO$_2$, respectively, and vary the carbon costs of consumer 3 between \$0/tCO$_2$ and \$25/tCO$_2$. 
The results  are shown in Table \ref{tab:carbon_pricing_results} and illustrated in Fig.\ref{fig:impactc}. 
Fig.\ref{fig:totaldemand} shows the total demand (in blue) and generation cost (in red)
as we change the carbon cost of consumer 3, while Fig.\ref{fig:Emissions} show the total and average carbon emissions. 
{\blue At carbon costs below $\$15$/MWh, we observe that the total generation cost first increases and then remains stable at \$350, while the total (and average) emissions first decrease and then remain stable at 29 tCO$_2$ (or 0.64 tCO$_2$/MWh) as the carbon cost of consumer 3 increases.} 
The load in this carbon cost range remains constant, indicating that the emission reductions are due to generation redispatch rather than load reduction. At carbon costs greater than $\$15$/MWh, consumer 3 lowers its consumption, leading to a reduction of both total generation cost and total carbon emissions. However, the average emissions increase as the remaining loads are served by less expensive, but more polluting generators.

{\blue Fig.\ref{fig:DemandAdjPrices}} shows the carbon-adjusted prices for each load. 
We observe that the loads have different carbon-adjusted prices due to their different carbon costs, and that the carbon-adjusted price is consistently highest for the consumer with the highest carbon cost\footnote{This is as expected since $\lambda_P$ is the same across all loads and generators in our system (due to the lack of congestion). The only differentiating factor between prices is thus the carbon-adjustment, which according to Theorem \ref{theoremodv} causes larger price increases for loads with higher carbon costs.}.
Specifically, consumers 2 and 3 with non-zero carbon cost pay a \emph{higher} price for their electricity compared to the standard market clearing, 
while consumer 1 with zero carbon cost experiences a \emph{reduction} in the (carbon-adjusted) electricity price from \$10/MWh to \$6/MWh. 
\textcolor{black}{This price reduction for the carbon-agnostic consumer happens because the carbon costs submitted by carbon-sensitive consumers reduces the demand for (low-cost) carbon-intense generation, leading to a system-wide redispatch where the marginal cost of serving the carbon-agnostic consumer is lower.}
The increase in total generation cost due to introducing carbon costs is thus primarily allocated to loads with non-zero carbon costs.




\subsubsection{\textcolor{black}{Multiple Optimal Solutions}}
We next compare multiple optimal solutions obtained from the carbon-aware model
for the case with carbon costs $c_D = [0,15,20]$. The results are shown in Table \ref{tab:model_comparison}. 
We observe that all solutions achieve the same total objective value and the same carbon-adjusted prices. However, the values for $P_G$ and $P_D$ are different, indicating that there are multiple optimal solutions with different splits between the different components of the
objective value, \textcolor{black}{i.e. different values of consumer utility, generation cost, and carbon cost. 
We further observe that the carbon-adjusted prices of several consumers and generators are equal to their respective utility or generation costs, indicating that there are several individual actors who have multiple optimal solutions in their individual profit maximization problems.}

\textcolor{black}{Specifically, in this case, the carbon-adjusted prices for consumer $3$ and generator $1$ are equal to their utility and generation cost coefficients,
meaning that the objective functions of their individual optimization problems are zero. Thus, any value of demand $d_3$ and generation $g_1$ within their respective bounds is optimal, meaning that \emph{both} demand $d_3$ and generation $g_1$ are marginal actors. Since we have two such marginal actors, any change $\Delta$ for which both the demand $P_{D,3}+\Delta$ and generation $P_{G,1}+\Delta$ remain within the bounds represents an optimal solution.}

\textcolor{black}{In a standard carbon-agnostic electricity market without transmission constraints, we have at most one electricity price and one marginal generator or load (unless there are multiple generators or loads whose cost coefficient matches the electricity price). In contrast, in the carbon-aware electricity market, there are multiple carbon-adjusted prices. This provides an opportunity for multiple marginal generators and loads, and thus more frequent presence of multiple solutions. The bolded carbon-adjusted prices in Table \ref{tab:carbon_pricing_results}, which highlight cases where the carbon-adjusted prices equal to the demand utility or generation cost, demonstrate the relatively frequent presence of having both marginal loads and marginal generators in the optimal solution. Stated differently, the introduction of a new term $c_D^\intercal E_D$ in the cost function of the centralized optimization problem provides more opportunities to trade off between consumer utility, carbon costs and generation cost compared to a standard electricity market.}

\subsection{Extension to RTS-GMLC System}
Next, we examine how our model impacts market clearing outcomes in the RTS-GMLC system \cite{barrows2019ieee} 
with 73 buses, 120 branches, 158 generators, and 51 loads. 
We assign emission factors for each generator based on their assigned fuel type and data from the US Department of Energy \cite{emissiondata}. Specifically, we assign $e_G=\{0.6042, 0.7434, 0.9606\}$ for natural gas, oil and coal generators, respectively, and assume that solar, wind and hydro have $e_G=0$.
The generation costs lie in the range of $[0,74.64]$\$/MWh, where the renewable generators have a cost of 0\$/MWh. 
We draw consumer utility values $u_D$ from a uniform distribution in the range of $[20,80]$~\$/MWh.
We assign consumer carbon costs $c_D$ by first drawing values from a uniform distribution in a range of $[10, 30]$ \$/MWh, and then randomly assigning zero carbon costs to 25\% of the consumers to simulate carbon-agnostic consumers. 
We further set $P_D^{\min}=0.8\times P_d$ and $P_D^{\max}=1.2\times P_d$, where $P_d$ is the baseline load value, 
\textcolor{black}{and we model congestion by incorporating transmission line limits to the RTS-GMLC system.}

\begin{figure}[t]
  \centering
  \begin{subfigure}{0.49\linewidth}
    \centering
    \includegraphics[width=\linewidth]{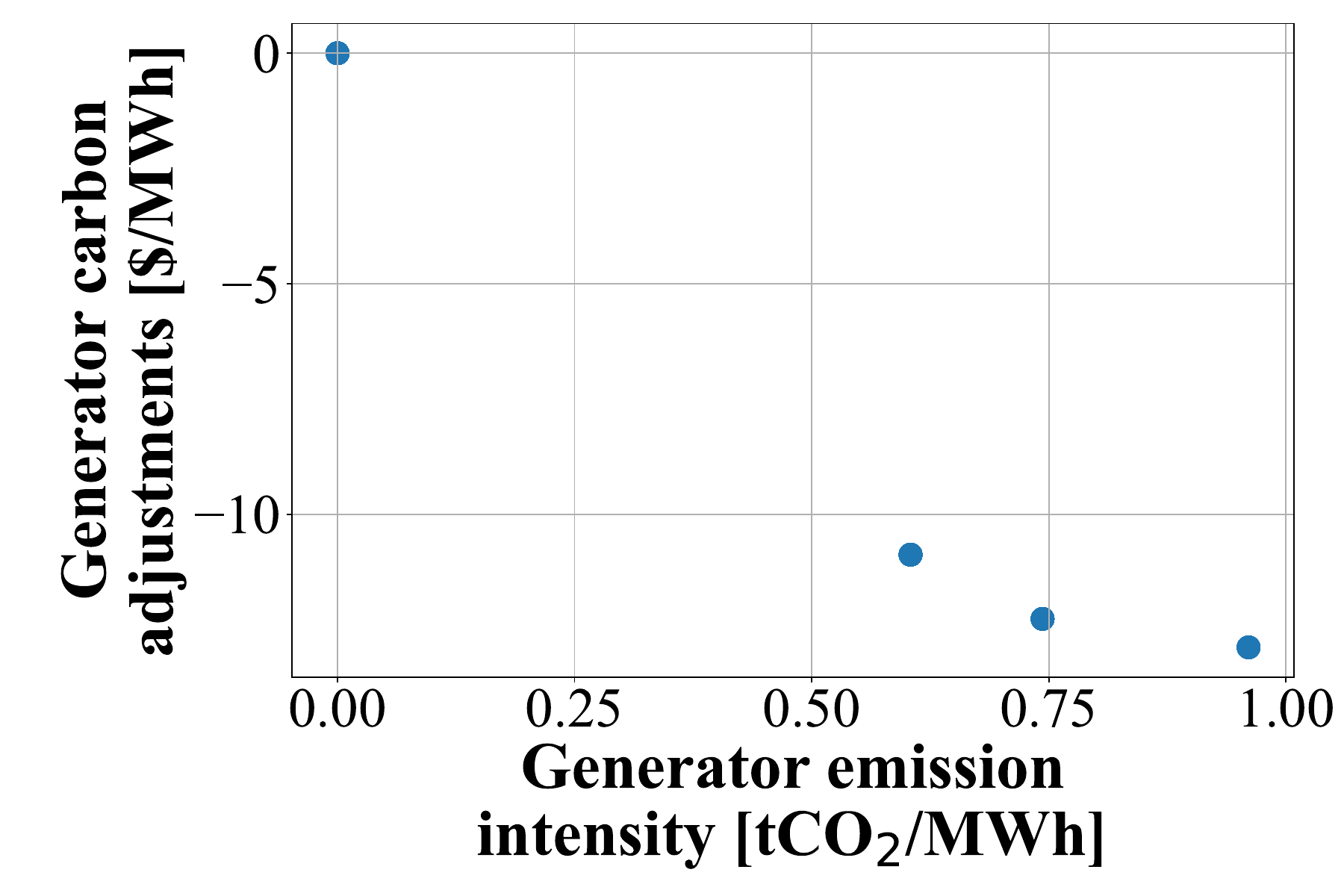}
    \caption{Generator $\lambda_G$ vs $e_G$.}
    \label{fig:GenAdj-vs-Em}
  \end{subfigure}
    \hfill
  \begin{subfigure}{0.49\linewidth}
    \centering
    \includegraphics[width=\linewidth]{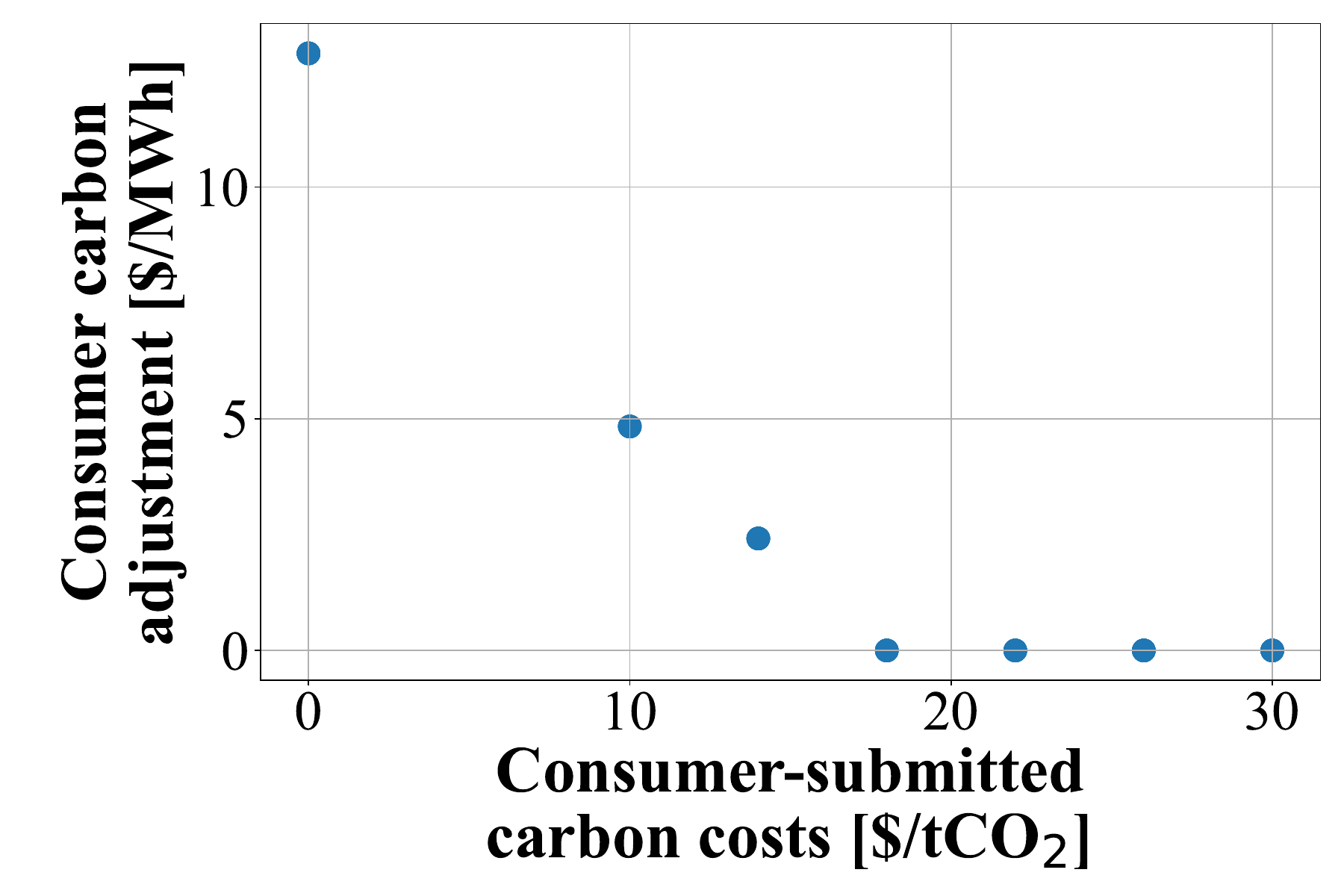}
    \caption{Consumer $\lambda_D$ vs $c_D$.}
    \label{fig:DemandAdj-vs-Cd}
  \end{subfigure}

  \caption{\small Relationship between the generator carbon-adjustments $\lambda_G$ and emission factors $e_G$ (left) and consumer carbon adjustments $\lambda_D$ and carbon costs $c_D$ (right).}
  \label{fig:Price-Adjustments}
\end{figure}
\begin{figure}[t]
  \centering
  \begin{subfigure}{0.49\linewidth}
    \centering
    \includegraphics[width=\linewidth]{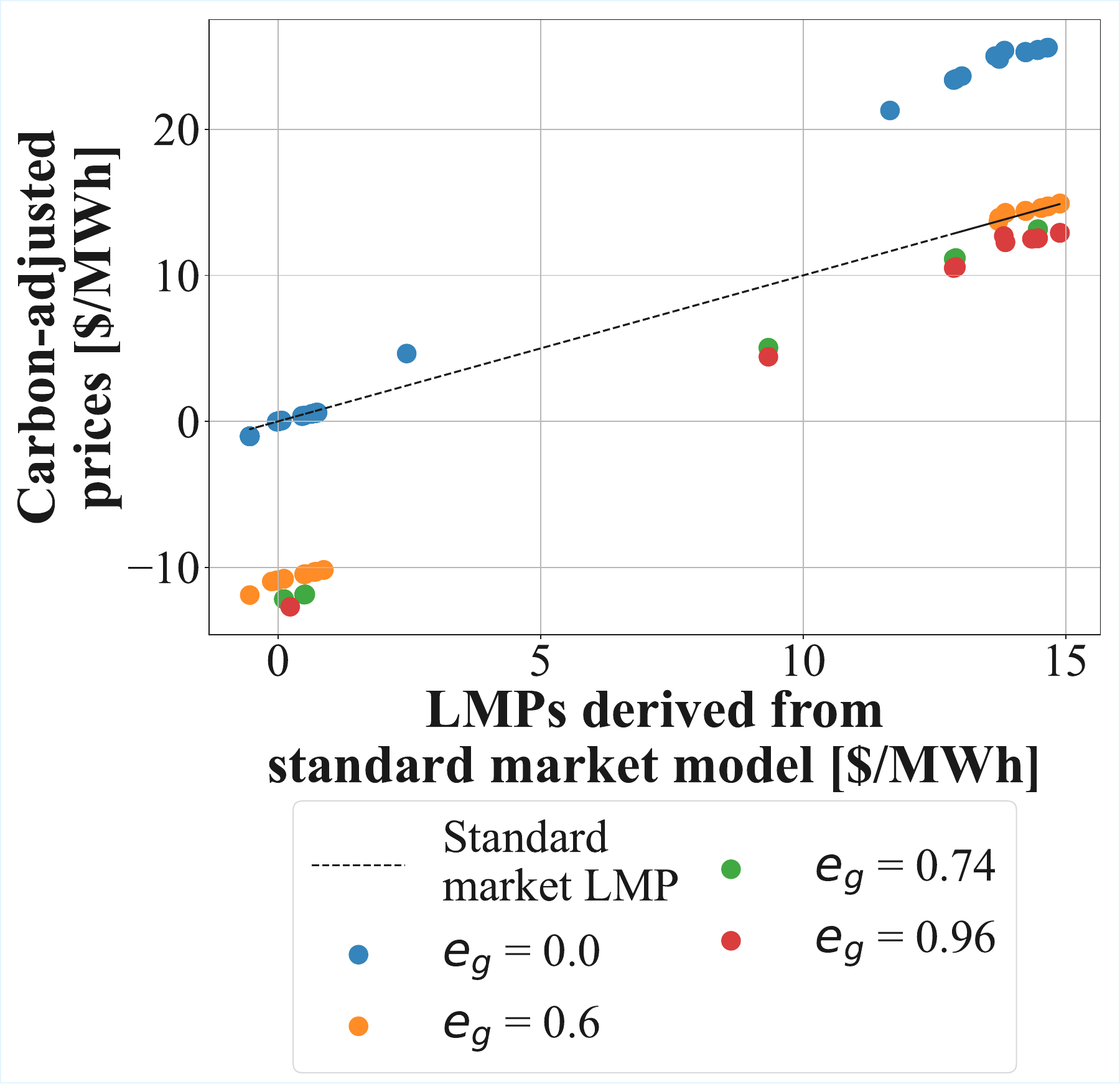}
    \caption{Generator prices}
    \label{fig:GenAdjPrices}
  \end{subfigure}
    \hfill
  \begin{subfigure}{0.49\linewidth}
    \centering
    \includegraphics[width=\linewidth]{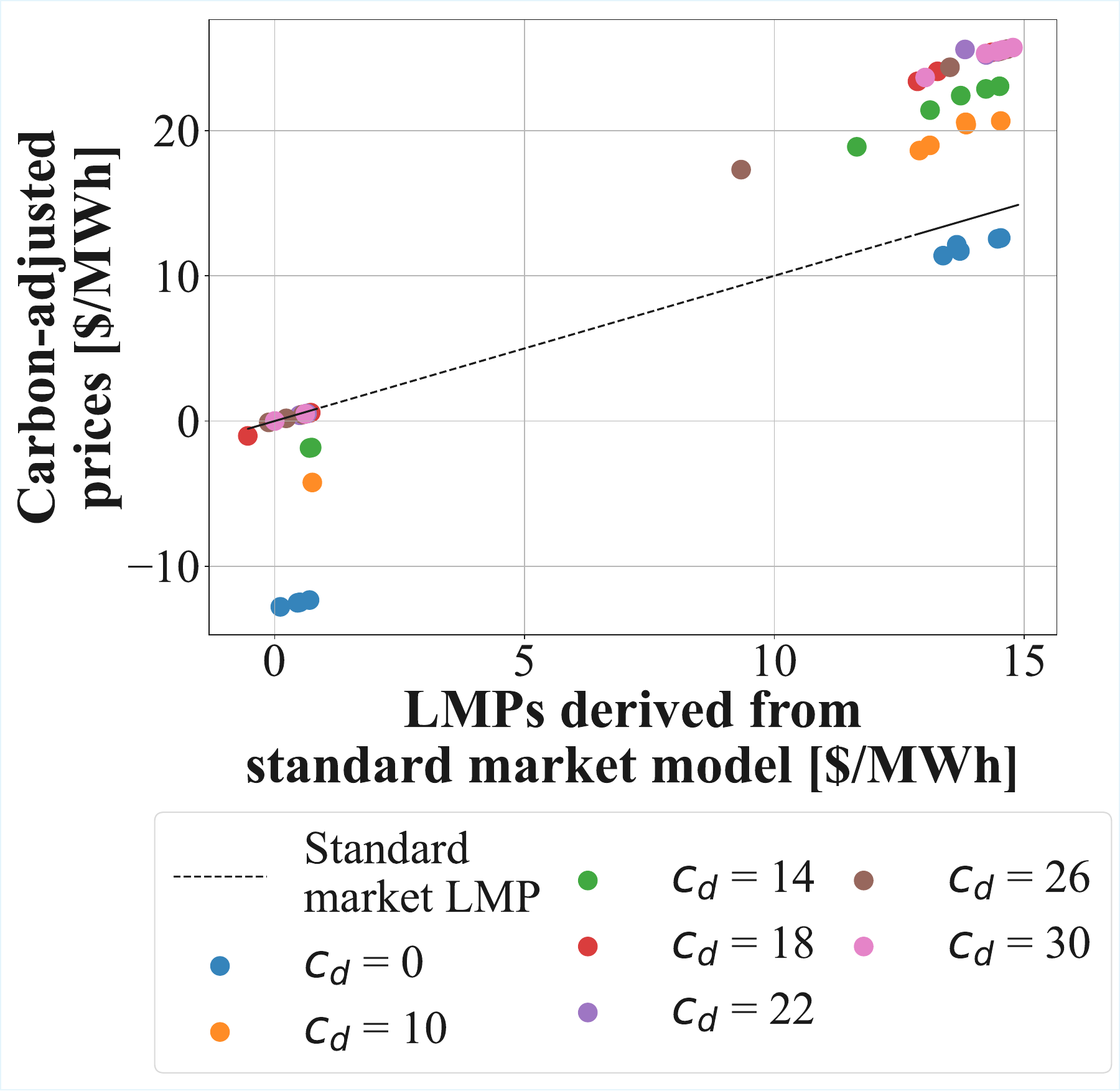}
    \caption{Consumer prices}
    \label{fig:DemandAdjPrices2}
  \end{subfigure}
  \caption{\small Comparing carbon-adjusted and carbon-agnostic prices.}
  \label{fig:AdjustedPrices}
\end{figure}

\subsubsection{Ordering of $\lambda_{G}$ and $\lambda_D$}
\label{suborder}
We first evaluate the relationship between generator carbon adjustments $\lambda_G$ and emission factors $e_G$, shown in 
Fig.~\ref{fig:GenAdj-vs-Em}.
From Fig.~\ref{fig:GenAdj-vs-Em}, we observe that generator carbon adjustments $\lambda_G$ decrease as the emission intensities $e_G$ increase. 
Further, generators with the same emission intensities $e_G$ have the same carbon adjustment $\lambda_G$, i.e. there are only four points corresponding to each generator type in Fig.~\ref{fig:GenAdj-vs-Em}, as expected from Corollary IV.3. 

We next assess the relationship between consumer carbon costs $c_D$ and carbon adjustments $\lambda_D$, shown in Fig.~\ref{fig:DemandAdj-vs-Cd}. We observe that carbon adjustments $\lambda_D$ decrease as the carbon costs $c_D$ increase, until $c_D$ reach \$18/tCO$_2$. For carbon costs $c_D\geq$~\$18/tCO$_2$, the carbon adjustments $\lambda_D$ become equal to zero. This is because these consumers are all served by renewable generators with $e_G = 0$, and thus have the same carbon adjustments $\lambda_D$ as implied by Corollary IV.4. 

\subsubsection{Carbon-Adjusted Prices}
The carbon adjustments impact the prices for generators and consumers, with higher adjustments suggesting higher payments (for generators) or lower electricity costs (for consumers). However, the electricity price is also impacted by congestion in the system, which is reflected in variations of $\lambda_P$. \textcolor{black}{To provide an example of how carbon costs impact a congested system,} we next consider how the carbon-adjusted prices $\lambda_P+\lambda_G$ for generators and $\lambda_P+\lambda_D$ compare with LMPs obtained from a standard market clearing. 

Fig. \ref{fig:AdjustedPrices} plots the carbon-adjusted prices for generators (left) and consumers (right) against the LMPs from the standard market clearing. {\blue The horizontal and vertical coordinates of each circle represent the LMP and carbon-adjusted price, respectively. Circles above (or below) the black line indicate generators or loads that experience higher (or lower) prices after carbon adjustment.
Further, Fig. \ref{fig:RTS-GMLC} shows where the congested lines are located (in red), as well as the geographical variation in differences between carbon-adjusted electricity prices and standard (carbon-agnostic) LMPs for different consumers and generators. 


From these figures, we observe that the prices are clustered into a lower cost group (generators and consumers with lower LMPs in Fig. \ref{fig:AdjustedPrices} and light blue circles in Fig. \ref{fig:RTS-GMLC}) and higher cost group (generators and consumers with higher LMPs in Fig. \ref{fig:AdjustedPrices} and dark blue circles in Fig. \ref{fig:RTS-GMLC}). As expected, within each range, generators with lower emissions receive higher prices (and thus larger payments) while consumers with higher carbon costs have to pay higher prices. However, the relative change in price compared to the LMPs from the standard market clearing is different in each range. From Fig. \ref{fig:AdjustedPrices}, we observe that for generators and consumers connected to nodes with a lower price, the carbon-adjusted price tends to be lower than the standard LMP. In contrast, generators and consumers connected at nodes in the high price range tend to experience carbon-adjusted prices that are higher than the standard LMP. This suggests that the consideration of consumer carbon costs has increased the impact of congestion compared to the carbon-agnostic case in this example. 


\begin{figure}[t]
    \centering
    \includegraphics[width=0.95\linewidth]{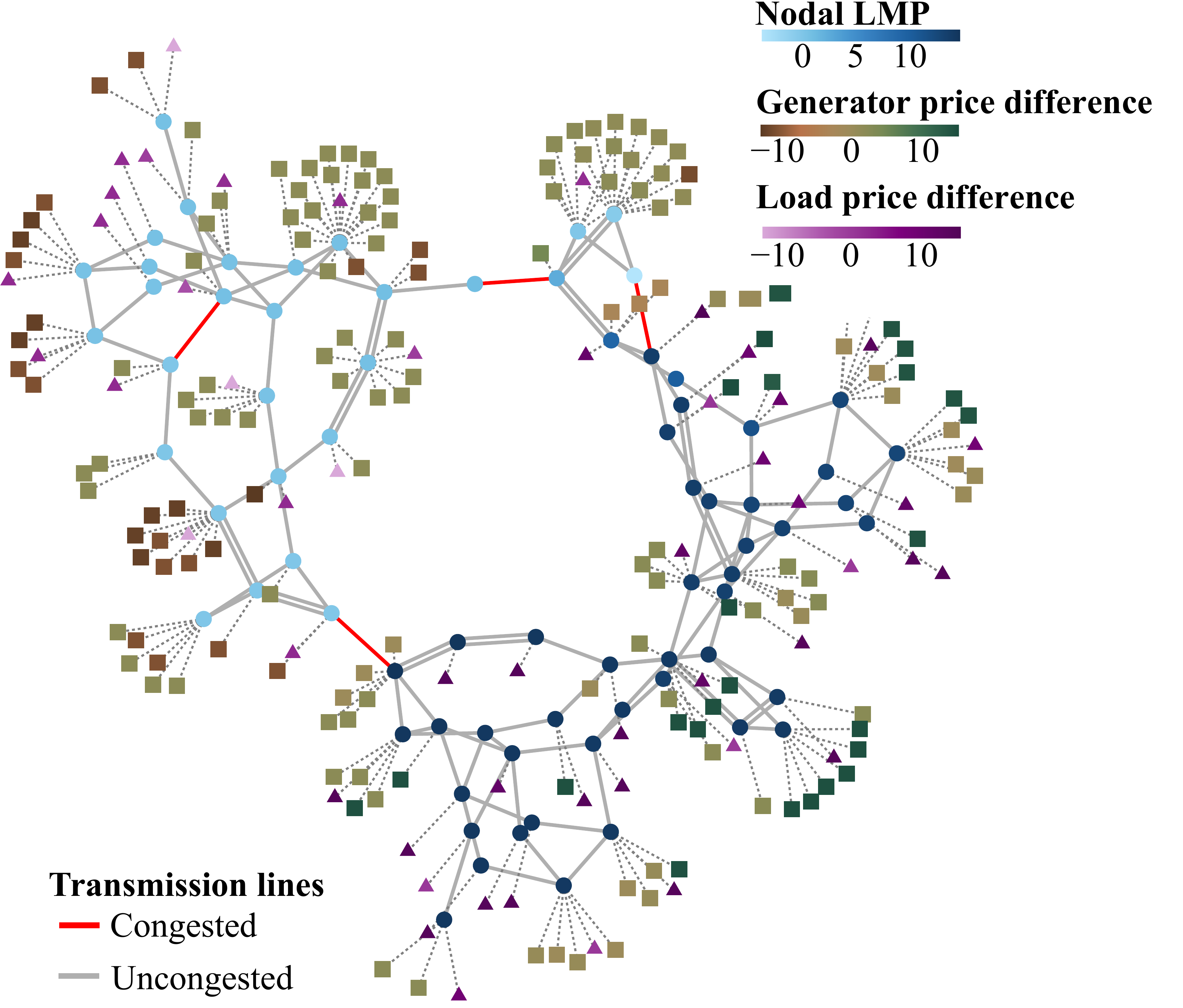}
    \caption{{\blue \small Changes in electricity prices with the introduction of carbon cost. Each circle represents a bus node, with color intensity representing standard (carbon-agnostic) LMPs. The squares connected to each node represent generators, where the color intensity reflects differences between generators' carbon-adjusted prices and standard LMP. Similarly, the triangles attached to each node represent consumers, with color intensity corresponding to their respective differences between consumers' carbon-adjusted prices and standard LMP. Congested transmission lines in carbon-aware case are highlighted by the red color (the uncongested lines are in gray).}}
    \label{fig:RTS-GMLC}
\end{figure}

To further understand this observation, we investigate the set of congested lines and the value of the dual variables associated with the corresponding transmission constraints in each market clearing case. In the standard (carbon-agnostic) electricity market, the congested lines are 119, 85, and 30 with dual variables at -16.6, -1.1, and 24.5. In the market clearing with carbon costs, the congested lines are 119, 85, 30 and 118 with dual variables $\overline{\eta}_{L,ij}$ or $\underline{\eta}_{L,ij}$ at -27.1, -0.9, 45.9 and -2.7, respectively. These findings demonstrate that the impact of congestion on the total operating cost increases when carbon cost is introduced, as indicated by an additional congested line and larger absolute values of the dual variables (with most lines). 
This is not unexpected, as our model further incentivizes the use of low cost renewable generation that is also low-carbon.
} 



{\blue \subsubsection{The Impact of Increasing Carbon Costs}
To investigate the impacts of increasing carbon costs on market clearing results, we keep the same setting as before but change the carbon costs by using different scale factors from 1 to 3. The comparison results are shown in Table \ref{tabIncrcarbon}, with the percentages in brackets indicating the reduction relative to the baseline (unscaled) case in the first row. We observe that with increasing carbon costs, Total Carbon, Total Generation Cost, and Total Generation decrease. 
The decrease in total generation is due to a reduction in demand, as increasing carbon costs cause more consumers to consume at the lowest demands.  
However, the Total Generation decreases more slowly compared to Total Carbon (and Total Generation Cost), suggesting that the average carbon emissions (i.e. emissions per MWh) are reduced as the carbon cost increases. 
}

\begin{table}[t]\footnotesize
\renewcommand{\arraystretch}{1.1}
\setlength{\tabcolsep}{3.2pt}
\centering
\caption{The Impacts of Increasing Carbon Costs on Market Clearing Results.}
\label{tabIncrcarbon}
\begin{tabular}{l|l|l|l}
\hline
Scale Factors& \begin{tabular}[c]{@{}c@{}}Total Carbon\\~[tCO$_2$]\end{tabular}  &\begin{tabular}[c]{@{}c@{}}Total Generation \\ Cost [\$]\end{tabular} &\begin{tabular}[c]{@{}c@{}}Total Generation\\~[MWh]\end{tabular}\\
\hline
\hline
Baseline &3592.2	&74189.3 &10230.4	\\
x 1.5	&3475.8 (-3.24\%)	&71945.1 (-3.02\%)	&10062.4 (-1.64\%)	\\
x 2	&3258.8 (-9.28\%)	&67728.7 (-8.71\%)	&~9746.7 (-4.73\%)	\\
x 2.5	&2543.5 (-29.19\%)	&58479.3 (-21.18\%)	&~8950.6 (-12.51\%)	\\
x 3	&2371.0 (-34.0\%)	&56592.9 (-23.72\%)	&~8762.6 (-14.35\%)	\\
\hline
\end{tabular}
\end{table}


\begin{table}
\footnotesize
\renewcommand{\arraystretch}{1.1}
\setlength{\tabcolsep}{3.2pt}
\centering
\caption{Yearly Average Metrics of Standard and Carbon-Aware Markets.}
\label{tab11}
\begin{tabular}{l|l|l|l}
\hline
Cases& \begin{tabular}[c]{@{}c@{}}Total Carbon\\~[tCO$_2$]\end{tabular}  &\begin{tabular}[c]{@{}c@{}}Total Generation \\ Cost [\$]\end{tabular} &\begin{tabular}[c]{@{}c@{}}Total Generation\\~[MWh]\end{tabular}\\
\hline
\hline
Standard&2850.78	&64489.33&	7653.21\\
Carbon-Aware	&2810.03 (-1.43\%)	&63608.14 (-1.37\%)	&7598.42 (-0.72\%)	\\
\hline
\end{tabular}
\end{table}

\begin{figure*}
    \centering
    \includegraphics[width=\linewidth]{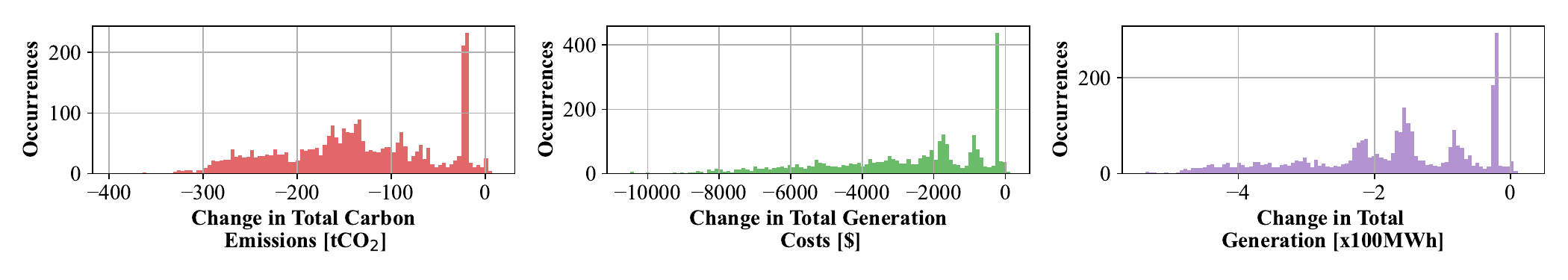}
    \caption{{\blue \small Changes of Total carbon emissions, Total generation cost, and Total generation using yearly data for the RTS-GMLC system.}}
    \label{fig:YearAnalysis}
\end{figure*}

{\blue \subsubsection{Yearly Data Analysis}

One year’s worth of hourly load and generation data is available for the RTS-GMLC system, totaling 8,784 cases. We assign consumer carbon costs $c_D$ for each consumer by first drawing values from a uniform distribution in a range of $[20, 60]$ \$/MWh, then randomly setting zero carbon costs to 25\% of the consumers to simulate carbon-agnostic consumers. These costs remain fixed for all hours. We then use the day-ahead load and renewable energy data from the RTS-GMLC data to calculate total carbon emissions, generation cost, and generation setpoints on an hourly basis. Across the full year, we observe 2,993 hours where our carbon-aware market clearing differs from the standard (carbon-agnostic) market clearing.

Table \ref{tab11} shows the yearly average metrics, calculated across all 8,764 hours, with the
percentages in brackets indicating the reduction relative to the
standard case. We observe that total carbon emissions, generation cost, and generation are reduced compared with the standard market clearing. The reduction in total generation indicates that overall demand has decreased due to the introduction of carbon costs. However, the percentage reduction in emissions and cost is larger than the percentage reduction in generation. This shows that the emission and cost reductions are not only a result of reduced demand, but that the introduction of carbon costs has also shifted generation from more polluting (and more expensive) to less polluting (and less expensive) generators. 

To investigate these effects in more detail, Fig. \ref{fig:YearAnalysis} shows the distribution of the hourly changes in emissions, generation cost and total generation (omitting the hours for which there are no changes).
While these figures demonstrate that there is generally a reduction in all metrics, Fig. \ref{fig:YearAnalysis} also shows that each metric may increase in specific instances.  
In particular, carbon emissions may sometimes increase in the carbon-aware market clearing compared with the standard market clearing. From a closer look at the results, we observe that this happens due to two main factors.
In some cases, the carbon-aware market clearing sometimes leads to an increase in system load, which raises total generation and emissions. In other cases, lower electricity prices for carbon-agnostic consumers relative to the standard market (as observed in Fig. \ref{fig:DemandAdjPrices2}) incentivize those consumers to increase their consumption, which is then met by carbon-intensive generators.}

\section{Conclusions}
\label{sec5}
In this paper, we analyze a recently proposed model to incorporate carbon allocation and consumer-side carbon costs into electricity market clearing. We derive an equivalent equilibrium formulation that gives rise to carbon-adjusted electricity prices and prove that the proposed carbon-cost based electricity market model satisfies similar market properties as current markets based on locational marginal pricing. Further, we show both theoretically and numerically that the proposed market clearing rewards low-emitting generators with higher electricity prices, and that increases in generation costs (compared with standard market clearing) are primarily allocated to consumers with higher carbon costs. 

This paper provides several avenues for future work. 
First, the proposed carbon allocation mechanism directly assigns power from generators to consumers while neglecting the physical characteristics of the electric grid and existing contractual agreements. As an extension, we plan to analyze the impacts of physical network constraints on power delivery and how existing power purchase agreements impact the carbon allocation. 
\textcolor{black}{
We are also interested in considering electricity market formulations with more complex constraints, such as those arising from unit commitment and the consideration of AC power flow. Extending the single optimization formulation is conceptually straightforward, and can be achieved by replacing the current DC OPF constraints with their corresponding AC power flow or unit commitment counterparts. However, more research is needed to extend the equilibrium analysis, as the introduction of non-convexities complicates the derivation and interpretation of dual variables. 
Further, formulations with multiple time steps are of interest as they allow for modeling of carbon-aware load shifting, where loads adapt their electricity use across time while keeping the overall consumption constant.} 
Finally, we plan to extend our analysis by investigating potential market manipulation opportunities, particularly in the elicitation and verification of carbon costs, further analysis of how costs and profits are allocated among participants,  \textcolor{black}{and exploring in greater detail the conditions that give rise to multiple optimal solutions.}

\bibliographystyle{IEEEtran}
\bibliography{ref}

\end{document}